\documentclass[10pt,twoside,twocolumn]{IEEEtran}

\usepackage{cite}
\usepackage[T1]{fontenc}
\usepackage{graphicx}
\usepackage{amssymb}
\usepackage{amsmath}
\usepackage{amsthm}
\usepackage{subfigure}
\usepackage{booktabs} 
\usepackage{multirow}
\usepackage{microtype}
\usepackage{balance}
\usepackage{xcolor}
\usepackage{xfrac}    
\usepackage{algorithm}
\usepackage{setspace}
\usepackage{mdframed}
\usepackage{tikz}
\usepackage{circledsteps}
\usepackage{dblfloatfix}

\usepackage{algorithmicx}
\usepackage{algpseudocode}
\algrenewcommand\algorithmicindent{0.7em}%
\usepackage{xpatch}

\usepackage[hidelinks]{hyperref}

\makeatletter
\newcommand\fs@spaceruled{\def\@fs@cfont{\bfseries}\let\@fs@capt\floatc@ruled
  \def\@fs@pre{\vspace{0.1\baselineskip}\hrule height.8pt depth0pt \kern2pt}%
  \def\@fs@post{\kern2pt\hrule\relax}%
  \def\@fs@mid{\kern2pt\hrule\kern2pt}%
  \let\@fs@iftopcapt\iftrue}
\makeatother

\definecolor{myred}{HTML}{9E292B}
\definecolor{myblue}{HTML}{235787}
\definecolor{mygreen}{HTML}{5E6638}
\definecolor{mygray}{HTML}{444444}
\definecolor{myblack}{HTML}{000000}
\definecolor{mywhite}{HTML}{FFFFFF}

\definecolor{myaltred}{HTML}{D46A78}
\definecolor{myaltblue}{HTML}{6699C2}
\definecolor{myaltgreen}{HTML}{B0B58C}
\definecolor{myaltgray}{HTML}{AAAAAA}

\definecolor{mylightred1}{HTML}{B15455}
\definecolor{mylightred2}{HTML}{C57F80}
\definecolor{mylightred3}{HTML}{D8A9AA}
\definecolor{mylightred4}{HTML}{ECD4D5}
\definecolor{mylightblue1}{HTML}{5A7DA5}
\definecolor{mylightblue2}{HTML}{7D99BA}
\definecolor{mylightblue3}{HTML}{B3C3D7}
\definecolor{mylightblue4}{HTML}{D3DCE8}
\definecolor{mydarkgreen}{HTML}{3E4822}
\definecolor{mylightgreen1}{HTML}{828859}
\definecolor{mylightgreen2}{HTML}{9AA075}
\definecolor{mylightgreen3}{HTML}{B8BC96}
\definecolor{mylightgreen4}{HTML}{D4D4B8}
\definecolor{mylightgray1}{HTML}{6F6F6F}
\definecolor{mylightgray2}{HTML}{999999}
\definecolor{mylightgray3}{HTML}{B4B4B4}
\definecolor{mylightgray4}{HTML}{DCDCDC}

\usepackage{amssymb}
\usepackage{amsfonts}
\usepackage{mathrsfs}
\usepackage{xspace}
\usepackage{bm}
\usepackage{upgreek}

\newcommand{\safemath}[2]{\newcommand{#1}{\ensuremath{#2}\xspace}}

\safemath{\bma}{\mathbf{a}}
\safemath{\bmb}{\mathbf{b}}
\safemath{\bmc}{\mathbf{c}}
\safemath{\bmd}{\mathbf{d}}
\safemath{\bme}{\mathbf{e}}
\safemath{\bmf}{\mathbf{f}}
\safemath{\bmg}{\mathbf{g}}
\safemath{\bmh}{\mathbf{h}}
\safemath{\bmi}{\mathbf{i}}
\safemath{\bmj}{\mathbf{j}}
\safemath{\bmk}{\mathbf{k}}
\safemath{\bml}{\mathbf{l}}
\safemath{\bmm}{\mathbf{m}}
\safemath{\bmn}{\mathbf{n}}
\safemath{\bmo}{\mathbf{o}}
\safemath{\bmp}{\mathbf{p}}
\safemath{\bmq}{\mathbf{q}}
\safemath{\bmr}{\mathbf{r}}
\safemath{\bms}{\mathbf{s}}
\safemath{\bmt}{\mathbf{t}}
\safemath{\bmu}{\mathbf{u}}
\safemath{\bmv}{\mathbf{v}}
\safemath{\bmw}{\mathbf{w}}
\safemath{\bmx}{\mathbf{x}}
\safemath{\bmy}{\mathbf{y}}
\safemath{\bmz}{\mathbf{z}}
\safemath{\bmzero}{\mathbf{0}}
\safemath{\bmone}{\mathbf{1}}

\bmdefine{\biad}{a}
\bmdefine{\bibd}{b}
\bmdefine{\bicd}{c}
\bmdefine{\bidd}{d}
\bmdefine{\bied}{e}
\bmdefine{\bifd}{f}
\bmdefine{\bigd}{g}
\bmdefine{\bihd}{h}
\bmdefine{\biid}{i}
\bmdefine{\bijd}{j}
\bmdefine{\bikd}{k}
\bmdefine{\bild}{l}
\bmdefine{\bimd}{m}
\bmdefine{\bind}{n}
\bmdefine{\biod}{o}
\bmdefine{\bipd}{p}
\bmdefine{\biqd}{q}
\bmdefine{\bird}{r}
\bmdefine{\bisd}{s}
\bmdefine{\bitd}{t}
\bmdefine{\biud}{u}
\bmdefine{\bivd}{v}
\bmdefine{\biwd}{w}
\bmdefine{\bixd}{x}
\bmdefine{\biyd}{y}
\bmdefine{\bizd}{z}

\bmdefine{\bixid}{\xi}
\bmdefine{\bilambdad}{\lambda}
\bmdefine{\bimud}{\mu}
\bmdefine{\bithetad}{\theta}
\bmdefine{\biphid}{\phi}
\bmdefine{\bideltad}{\delta}

\safemath{\bmia}{\biad}
\safemath{\bmib}{\bibd}
\safemath{\bmic}{\bicd}
\safemath{\bmid}{\bidd}
\safemath{\bmie}{\bied}
\safemath{\bmif}{\bifd}
\safemath{\bmig}{\bigd}
\safemath{\bmih}{\bihd}
\safemath{\bmii}{\biid}
\safemath{\bmij}{\bijd}
\safemath{\bmik}{\bikd}
\safemath{\bmil}{\bild}
\safemath{\bmim}{\bimd}
\safemath{\bmin}{\bind}
\safemath{\bmio}{\biod}
\safemath{\bmip}{\bipd}
\safemath{\bmiq}{\biqd}
\safemath{\bmir}{\bird}
\safemath{\bmis}{\bisd}
\safemath{\bmit}{\bitd}
\safemath{\bmiu}{\biud}
\safemath{\bmiv}{\bivd}
\safemath{\bmiw}{\biwd}
\safemath{\bmix}{\bixd}
\safemath{\bmiy}{\biyd}
\safemath{\bmiz}{\bizd}

\safemath{\bmxi}{\bixid}
\safemath{\bmlambda}{\bilambdad}
\safemath{\bmmu}{\bimud}
\safemath{\bmtheta}{\bithetad}
\safemath{\bmphi}{\biphid}
\safemath{\bmdelta}{\bideltad}

\safemath{\bA}{\mathbf{A}}
\safemath{\bB}{\mathbf{B}}
\safemath{\bC}{\mathbf{C}}
\safemath{\bD}{\mathbf{D}}
\safemath{\bE}{\mathbf{E}}
\safemath{\bF}{\mathbf{F}}
\safemath{\bG}{\mathbf{G}}
\safemath{\bH}{\mathbf{H}}
\safemath{\bI}{\mathbf{I}}
\safemath{\bJ}{\mathbf{J}}
\safemath{\bK}{\mathbf{K}}
\safemath{\bL}{\mathbf{L}}
\safemath{\bM}{\mathbf{M}}
\safemath{\bN}{\mathbf{N}}
\safemath{\bO}{\mathbf{O}}
\safemath{\bP}{\mathbf{P}}
\safemath{\bQ}{\mathbf{Q}}
\safemath{\bR}{\mathbf{R}}
\safemath{\bS}{\mathbf{S}}
\safemath{\bT}{\mathbf{T}}
\safemath{\bU}{\mathbf{U}}
\safemath{\bV}{\mathbf{V}}
\safemath{\bW}{\mathbf{W}}
\safemath{\bX}{\mathbf{X}}
\safemath{\bY}{\mathbf{Y}}
\safemath{\bZ}{\mathbf{Z}}

\safemath{\bZero}{\mathbf{0}}
\safemath{\bOne}{\mathbf{1}}
\safemath{\bDelta}{\mathbf{\Delta}}
\safemath{\bLambda}{\mathbf{\UpLambda}}
\safemath{\bPhi}{\mathbf{\Upphi}}
\safemath{\bSigma}{\mathbf{\Upsigma}}
\safemath{\bOmega}{\mathbf{\Upomega}}
\safemath{\bTheta}{\mathbf{\Uptheta}}

\bmdefine{\biAd}{A}
\bmdefine{\biBd}{B}
\bmdefine{\biCd}{C}
\bmdefine{\biDd}{D}
\bmdefine{\biEd}{E}
\bmdefine{\biFd}{F}
\bmdefine{\biGd}{G}
\bmdefine{\biHd}{H}
\bmdefine{\biId}{I}
\bmdefine{\biJd}{J}
\bmdefine{\biKd}{K}
\bmdefine{\biLd}{L}
\bmdefine{\biMd}{M}
\bmdefine{\biOd}{N}
\bmdefine{\biPd}{O}
\bmdefine{\biQd}{P}
\bmdefine{\biRd}{R}
\bmdefine{\biSd}{S}
\bmdefine{\biTd}{T}
\bmdefine{\biUd}{U}
\bmdefine{\biVd}{V}
\bmdefine{\biWd}{W}
\bmdefine{\biXd}{X}
\bmdefine{\biYd}{Y}
\bmdefine{\biZd}{Z}

\bmdefine{\biDelta}{\Delta}
\bmdefine{\biLambda}{\Lambda}
\bmdefine{\biPhi}{\Phi}
\bmdefine{\biSigma}{\Sigma}
\bmdefine{\biOmega}{\Omega}
\bmdefine{\biTheta}{\Theta}

\safemath{\bimA}{\biAd}
\safemath{\bimB}{\biBd}
\safemath{\bimC}{\biCd}
\safemath{\bimD}{\biDd}
\safemath{\bimE}{\biEd}
\safemath{\bimF}{\biFd}
\safemath{\bimG}{\biGd}
\safemath{\bimH}{\biHd}
\safemath{\bimI}{\biId}
\safemath{\bimJ}{\biJd}
\safemath{\bimK}{\biKd}
\safemath{\bimL}{\biLd}
\safemath{\bimM}{\biMd}
\safemath{\bimN}{\biNd}
\safemath{\bimO}{\biOd}
\safemath{\bimP}{\biPd}
\safemath{\bimQ}{\biQd}
\safemath{\bimR}{\biRd}
\safemath{\bimS}{\biSd}
\safemath{\bimT}{\biTd}
\safemath{\bimU}{\biUd}
\safemath{\bimV}{\biVd}
\safemath{\bimW}{\biWd}
\safemath{\bimX}{\biXd}
\safemath{\bimY}{\biYd}
\safemath{\bimZ}{\biZd}

\safemath{\bimDelta}{\biDelta}
\safemath{\bimLambda}{\biLambda}
\safemath{\bimPhi}{\biPhi}
\safemath{\bimSigma}{\biSigma}
\safemath{\bimOmega}{\biOmega}
\safemath{\bimTheta}{\biTheta}

\safemath{\setA}{\mathcal{A}}
\safemath{\setB}{\mathcal{B}}
\safemath{\setC}{\mathcal{C}}
\safemath{\setD}{\mathcal{D}}
\safemath{\setE}{\mathcal{E}}
\safemath{\setF}{\mathcal{F}}
\safemath{\setG}{\mathcal{G}}
\safemath{\setH}{\mathcal{H}}
\safemath{\setI}{\mathcal{I}}
\safemath{\setJ}{\mathcal{J}}
\safemath{\setK}{\mathcal{K}}
\safemath{\setL}{\mathcal{L}}
\safemath{\setM}{\mathcal{M}}
\safemath{\setN}{\mathcal{N}}
\safemath{\setO}{\mathcal{O}}
\safemath{\setP}{\mathcal{P}}
\safemath{\setQ}{\mathcal{Q}}
\safemath{\setR}{\mathcal{R}}
\safemath{\setS}{\mathcal{S}}
\safemath{\setT}{\mathcal{T}}
\safemath{\setU}{\mathcal{U}}
\safemath{\setV}{\mathcal{V}}
\safemath{\setW}{\mathcal{W}}
\safemath{\setX}{\mathcal{X}}
\safemath{\setY}{\mathcal{Y}}
\safemath{\setZ}{\mathcal{Z}}
\safemath{\emptySet}{\varnothing}

\safemath{\colA}{\mathscr{A}}
\safemath{\colB}{\mathscr{B}}
\safemath{\colC}{\mathscr{C}}
\safemath{\colD}{\mathscr{D}}
\safemath{\colE}{\mathscr{E}}
\safemath{\colF}{\mathscr{F}}
\safemath{\colG}{\mathscr{G}}
\safemath{\colH}{\mathscr{H}}
\safemath{\colI}{\mathscr{I}}
\safemath{\colJ}{\mathscr{J}}
\safemath{\colK}{\mathscr{K}}
\safemath{\colL}{\mathscr{L}}
\safemath{\colM}{\mathscr{M}}
\safemath{\colN}{\mathscr{N}}
\safemath{\colO}{\mathscr{O}}
\safemath{\colP}{\mathscr{P}}
\safemath{\colQ}{\mathscr{Q}}
\safemath{\colR}{\mathscr{R}}
\safemath{\colS}{\mathscr{S}}
\safemath{\colT}{\mathscr{T}}
\safemath{\colU}{\mathscr{U}}
\safemath{\colV}{\mathscr{V}}
\safemath{\colW}{\mathscr{W}}
\safemath{\colX}{\mathscr{X}}
\safemath{\colY}{\mathscr{Y}}
\safemath{\colZ}{\mathscr{Z}}

\safemath{\opA}{\mathbb{A}}
\safemath{\opB}{\mathbb{B}}
\safemath{\opC}{\mathbb{C}}
\safemath{\opD}{\mathbb{D}}
\safemath{\opE}{\mathbb{E}}
\safemath{\opF}{\mathbb{F}}
\safemath{\opG}{\mathbb{G}}
\safemath{\opH}{\mathbb{H}}
\safemath{\opI}{\mathbb{I}}
\safemath{\opJ}{\mathbb{J}}
\safemath{\opK}{\mathbb{K}}
\safemath{\opL}{\mathbb{L}}
\safemath{\opM}{\mathbb{M}}
\safemath{\opN}{\mathbb{N}}
\safemath{\opO}{\mathbb{O}}
\safemath{\opP}{\mathbb{P}}
\safemath{\opQ}{\mathbb{Q}}
\safemath{\opR}{\mathbb{R}}
\safemath{\opS}{\mathbb{S}}
\safemath{\opT}{\mathbb{T}}
\safemath{\opU}{\mathbb{U}}
\safemath{\opV}{\mathbb{V}}
\safemath{\opW}{\mathbb{W}}
\safemath{\opX}{\mathbb{X}}
\safemath{\opY}{\mathbb{Y}}
\safemath{\opZ}{\mathbb{Z}}
\safemath{\opZero}{\mathbb{O}}
\safemath{\identityop}{\opI}

\safemath{\veca}{\bma}
\safemath{\vecb}{\bmb}
\safemath{\vecc}{\bmc}
\safemath{\vecd}{\bmd}
\safemath{\vece}{\bme}
\safemath{\vecf}{\bmf}
\safemath{\vecg}{\bmg}
\safemath{\vech}{\bmh}
\safemath{\veci}{\bmi}
\safemath{\vecj}{\bmj}
\safemath{\veck}{\bmk}
\safemath{\vecl}{\bml}
\safemath{\vecm}{\bmm}
\safemath{\vecn}{\bmn}
\safemath{\veco}{\bmo}
\safemath{\vecp}{\bmp}
\safemath{\vecq}{\bmq}
\safemath{\vecr}{\bmr}
\safemath{\vecs}{\bms}
\safemath{\vect}{\bmt}
\safemath{\vecu}{\bmu}
\safemath{\vecv}{\bmv}
\safemath{\vecw}{\bmw}
\safemath{\vecx}{\bmx}
\safemath{\vecy}{\bmy}
\safemath{\vecz}{\bmz}

\safemath{\veczero}{\bmzero}
\safemath{\vecone}{\bmone}
\safemath{\vecxi}{\bmxi}
\safemath{\veclambda}{\bmlambda}
\safemath{\vecmu}{\bmmu}
\safemath{\vectheta}{\bmtheta}
\safemath{\vecphi}{\bmphi}
\safemath{\vecdelta}{\bmdelta}

\safemath{\matA}{\bA}
\safemath{\matB}{\bB}
\safemath{\matC}{\bC}
\safemath{\matD}{\bD}
\safemath{\matE}{\bE}
\safemath{\matF}{\bF}
\safemath{\matG}{\bG}
\safemath{\matH}{\bH}
\safemath{\matI}{\bI}
\safemath{\matJ}{\bJ}
\safemath{\matK}{\bK}
\safemath{\matL}{\bL}
\safemath{\matM}{\bM}
\safemath{\matN}{\bN}
\safemath{\matO}{\bO}
\safemath{\matP}{\bP}
\safemath{\matQ}{\bQ}
\safemath{\matR}{\bR}
\safemath{\matS}{\bS}
\safemath{\matT}{\bT}
\safemath{\matU}{\bU}
\safemath{\matV}{\bV}
\safemath{\matW}{\bW}
\safemath{\matX}{\bX}
\safemath{\matY}{\bY}
\safemath{\matZ}{\bZ}
\safemath{\matzero}{\bmzero}

\safemath{\matDelta}{\bDelta}
\safemath{\matLambda}{\bLambda}
\safemath{\matPhi}{\bPhi}
\safemath{\matSigma}{\bSigma}
\safemath{\matOmega}{\bOmega}
\safemath{\matTheta}{\bTheta}

\safemath{\matidentity}{\matI}
\safemath{\matone}{\matO}

\safemath{\rnda}{A}
\safemath{\rndb}{B}
\safemath{\rndc}{C}
\safemath{\rndd}{D}
\safemath{\rnde}{E}
\safemath{\rndf}{F}
\safemath{\rndg}{G}
\safemath{\rndh}{H}
\safemath{\rndi}{I}
\safemath{\rndj}{J}
\safemath{\rndk}{K}
\safemath{\rndl}{L}
\safemath{\rndm}{M}
\safemath{\rndn}{N}
\safemath{\rndo}{O}
\safemath{\rndp}{P}
\safemath{\rndq}{Q}
\safemath{\rndr}{R}
\safemath{\rnds}{S}
\safemath{\rndt}{T}
\safemath{\rndu}{U}
\safemath{\rndv}{V}
\safemath{\rndw}{W}
\safemath{\rndx}{X}
\safemath{\rndy}{Y}
\safemath{\rndz}{Z}

\safemath{\rveca}{\bimA}
\safemath{\rvecb}{\bimB}
\safemath{\rvecc}{\bimC}
\safemath{\rvecd}{\bimD}
\safemath{\rvece}{\bimE}
\safemath{\rvecf}{\bimF}
\safemath{\rvecg}{\bimG}
\safemath{\rvech}{\bimH}
\safemath{\rveci}{\bimI}
\safemath{\rvecj}{\bimJ}
\safemath{\rveck}{\bimK}
\safemath{\rvecl}{\bimL}
\safemath{\rvecm}{\bimM}
\safemath{\rvecn}{\bimN}
\safemath{\rveco}{\bomO}
\safemath{\rvecp}{\bimP}
\safemath{\rvecq}{\bimQ}
\safemath{\rvecr}{\bimR}
\safemath{\rvecs}{\bimS}
\safemath{\rvect}{\bimT}
\safemath{\rvecu}{\bimU}
\safemath{\rvecv}{\bimV}
\safemath{\rvecw}{\bimW}
\safemath{\rvecx}{\bimX}
\safemath{\rvecy}{\bimY}
\safemath{\rvecz}{\bimZ}

\safemath{\rvecxi}{\bmxi}
\safemath{\rveclambda}{\bmlambda}
\safemath{\rvecmu}{\bmmu}
\safemath{\rvectheta}{\bmtheta}
\safemath{\rvecphi}{\bmphi}

\safemath{\rmatA}{\bimA}
\safemath{\rmatB}{\bimB}
\safemath{\rmatC}{\bimC}
\safemath{\rmatD}{\bimD}
\safemath{\rmatE}{\bimE}
\safemath{\rmatF}{\bimF}
\safemath{\rmatG}{\bimG}
\safemath{\rmatH}{\bimH}
\safemath{\rmatI}{\bimI}
\safemath{\rmatJ}{\bimJ}
\safemath{\rmatK}{\bimK}
\safemath{\rmatL}{\bimL}
\safemath{\rmatM}{\bimM}
\safemath{\rmatN}{\bimN}
\safemath{\rmatO}{\bimO}
\safemath{\rmatP}{\bimP}
\safemath{\rmatQ}{\bimQ}
\safemath{\rmatR}{\bimR}
\safemath{\rmatS}{\bimS}
\safemath{\rmatT}{\bimT}
\safemath{\rmatU}{\bimU}
\safemath{\rmatV}{\bimV}
\safemath{\rmatW}{\bimW}
\safemath{\rmatX}{\bimX}
\safemath{\rmatY}{\bimY}
\safemath{\rmatZ}{\bimZ}

\safemath{\rmatDelta}{\bimDelta}
\safemath{\rmatLambda}{\bimLambda}
\safemath{\rmatPhi}{\bimPhi}
\safemath{\rmatSigma}{\bimSigma}
\safemath{\rmatOmega}{\bimOmega}
\safemath{\rmatTheta}{\bimTheta}

\usepackage{amssymb}
\usepackage{amsfonts}
\usepackage{mathrsfs}
\usepackage{xspace}
\usepackage{bm}
\usepackage{fancyref}
\usepackage{textcomp}

\usepackage{multirow}
\usepackage{stmaryrd}

\newenvironment{textbmatrix}{	\setlength{\arraycolsep}{2.5pt}%
								\big[\begin{matrix}}{\end{matrix}\big]%
								\raisebox{0.08ex}{\vphantom{M}}}

\def\be{\begin{equation}}
\def\ee{\end{equation}}
\def\een{\nonumber \end{equation}}
\def\mat{\begin{bmatrix}}
\def\emat{\end{bmatrix}}
\def\btm{\begin{textbmatrix}}
\def\etm{\end{textbmatrix}}

\def\ba#1\ea{\begin{align}#1\end{align}}
\def\bas#1\eas{\begin{align*}#1\end{align*}}
\def\bs#1\es{\begin{split}#1\end{split}}
\def\bg#1\eg{\begin{gather}#1\end{gather}}
\def\bml#1\eml{\begin{multline}#1\end{multline}}
\def\bi#1\ei{\begin{itemize}#1\end{itemize}}

\newcommand{\lefto}{\mathopen{}\left}

\DeclareMathOperator{\Tr}{\opT r}			
\DeclareMathOperator{\rank}{rank}			
\DeclareMathOperator*{\argmax}{arg\;max}		
\DeclareMathOperator{\kron}{\otimes}			
\DeclareMathOperator{\Prob}{\opP}			
\DeclareMathOperator{\Exop}{\opE}			

\newcommand{\orth}{\perp}					
\newcommand{\Ex}[2]{\ensuremath{\Exop_{#1}\lefto[#2\right]}} 	

\newcommand{\tp}[1]{\ensuremath{#1^{\text{T}}}} 		
\newcommand{\herm}[1]{\ensuremath{#1^{\text{H}}}} 	
\newcommand{\inv}[1]{\ensuremath{#1^{-1}}} 	
\newcommand{\pinv}[1]{\ensuremath{#1^{\dagger}}} 	

\safemath{\dirac}{\delta}					
\safemath{\krond}{\dirac}					

\safemath{\upto}{\uparrow}
\safemath{\downto}{\downarrow}
\safemath{\iu}{j}							
\safemath{\ev}{\lambda}						
\safemath{\hilseqspace}{l^{2}}				
\newcommand{\banachfunspace}[1]{\setL^{#1}}	
\safemath{\hilfunspace}{\banachfunspace{2}}	

\safemath{\SNR}{\textit{SNR}} 				
\safemath{\PAR}{\textit{PAR}} 				
\safemath{\No}{N_0}							
\safemath{\Es}{E_s}							
\safemath{\Eb}{E_b}							
\safemath{\EbNo}{\frac{\Eb}{\No}}
\safemath{\EsNo}{\frac{\Es}{\No}}

\DeclareMathOperator{\CHop}{\ensuremath{\opH}} 
\safemath{\tvir}{\rndh_{\CHop}}				
\safemath{\tvtf}{\rndl_{\CHop}}				
\safemath{\spf}{\rnds_{\CHop}}				
\safemath{\bff}{H_{\CHop}}					

\safemath{\ircf}{r_{h}}						
\safemath{\tftvcf}{r_{s}}					
\safemath{\tfcf}{r_{l}}						
\safemath{\bfcf}{r_{H}}						

\safemath{\tcorr}{c_h}						
\safemath{\scf}{c_{s}}						
\safemath{\tfcorr}{c_{l}}					
\safemath{\fcorr}{c_{H}}						

\safemath{\mi}{I}							
\safemath{\capacity}{C}						

\safemath{\normal}{\mathcal{N}}			
\safemath{\jpg}{\mathcal{CN}}			
\safemath{\mchain}{\leftrightarrow}		

\safemath{\dB}{\,\mathrm{dB}}
\safemath{\dBm}{\,\mathrm{dBm}}
\safemath{\Hz}{\,\mathrm{Hz}}
\safemath{\kHz}{\,\mathrm{kHz}}
\safemath{\MHz}{\,\mathrm{MHz}}
\safemath{\GHz}{\,\mathrm{GHz}}
\safemath{\s}{\,\mathrm{s}}
\safemath{\ms}{\,\mathrm{ms}}
\safemath{\mus}{\,\mathrm{\text{\textmu}s}}
\safemath{\ns}{\,\mathrm{ns}}
\safemath{\ps}{\,\mathrm{ps}}
\safemath{\meter}{\,\mathrm{m}}
\safemath{\mm}{\,\mathrm{mm}}
\safemath{\cm}{\,\mathrm{cm}}
\safemath{\m}{\,\mathrm{m}}
\safemath{\W}{\,\mathrm{W}}
\safemath{\mW}{\, \mathrm{mW}}
\safemath{\J}{\,\mathrm{J}}
\safemath{\K}{\,\mathrm{K}}
\safemath{\bit}{\,\mathrm{bit}}
\safemath{\nat}{\,\mathrm{nat}}

\safemath{\define}{\triangleq}			

\safemath{\equivalent}{\sim}
\safemath{\distas}{\sim}					
\safemath{\sdiff}{\Delta}				

\safemath{\reals}{\mathbb{R}}
\safemath{\positivereals}{\reals_{+}}
\safemath{\integers}{\mathbb{Z}}
\safemath{\posint}{\integers_{+}}
\safemath{\naturals}{\mathbb{N}}
\safemath{\posnaturals}{\naturals_{+}}
\safemath{\complexset}{\mathbb{C}}
\safemath{\rationals}{\mathbb{Q}}

\newcommand*{\fancyrefapplabelprefix}{app}		
\newcommand*{\fancyrefthmlabelprefix}{thm}		
\newcommand*{\fancyreflemlabelprefix}{lem}		
\newcommand*{\fancyrefcorlabelprefix}{cor}		
\newcommand*{\fancyrefdeflabelprefix}{def}		
\newcommand*{\fancyrefproplabelprefix}{prop}		
\newcommand*{\fancyrefexmpllabelprefix}{exmpl}
\newcommand*{\fancyrefalglabelprefix}{alg}		
\newcommand*{\fancyreftbllabelprefix}{tbl}		

\frefformat{vario}{\fancyrefseclabelprefix}{Sec.~#1}
\frefformat{vario}{\fancyrefthmlabelprefix}{Thm.~#1}
\frefformat{vario}{\fancyreftbllabelprefix}{Table~#1}
\frefformat{vario}{\fancyreflemlabelprefix}{Lemma~#1}
\frefformat{vario}{\fancyrefcorlabelprefix}{Corollary~#1}
\frefformat{vario}{\fancyrefdeflabelprefix}{Def.~#1}
\frefformat{vario}{\fancyreffiglabelprefix}{Fig.~#1}
\frefformat{vario}{\fancyrefapplabelprefix}{Appendix~#1}
\frefformat{vario}{\fancyrefeqlabelprefix}{(#1)}
\frefformat{vario}{\fancyrefproplabelprefix}{Prop.~#1}
\frefformat{vario}{\fancyrefexmpllabelprefix}{Example~#1}
\frefformat{vario}{\fancyrefalglabelprefix}{Alg.~#1}

 \newtheorem{thm}{Theorem}
 \newtheorem{cor}[thm]{Corollary}   
 
 \newtheorem{defi}{Definition}
 \newtheorem{lem}[thm]{Lemma}

 \newtheorem*{remark*}{Remark}

\safemath{\dictab}{[\,\dicta\,\,\dictb\,]}

\safemath{\ysig}{\bmy}
\safemath{\ysighat}{\hat{\ysig}}
\safemath{\ysigdim}{M}
\safemath{\xsig}{\bmx}
\safemath{\xsigdim}{N}
\safemath{\nx}{n_x}
\safemath{\zsig}{\bmz}
\safemath{\zsigdim}{\ysigdim}
\safemath{\rsig}{\bmr}
\safemath{\Adict}{\bA}
\safemath{\Adicttilde}{\widetilde{\Adict}}
\safemath{\Adictdim}{\outputdim\times\xsigdim}
\safemath{\avec}{\bma}
\safemath{\avectilde}{\tilde{\avec}}
\safemath{\Bdict}{\bB}
\safemath{\Bdicttilde}{\widetilde{\Bdict}}
\safemath{\Cdict}{\bC}
\safemath{\cvec}{\bmc}
\safemath{\Ddict}{\bD}
\safemath{\Ddictdim}{\ysigdim\times\xsigdim}
\safemath{\dvec}{\bmd}
\safemath{\Ddicttilde}{\widetilde{\bD}}
\safemath{\Bonb}{\bB}
\safemath{\bvec}{\bmb}
\safemath{\Bonbdim}{\ysigdim\times\ysigdim}
\safemath{\noise}{\bmn}
\safemath{\noisedim}{\ysigim}
\safemath{\err}{\bme}
\safemath{\errdim}{\ysigdim}
\safemath{\errset}{\setE}
\safemath{\nerr}{n_e}
\safemath{\delop}{\bP_\errset}
\safemath{\delopc}{\bP_{{\errset}^c}}

\safemath{\cplxi}{\imath}
\safemath{\cplxj}{\jmath}

\safemath{\dict}{\matD}
\safemath{\inputdim}{N}		
\safemath{\outputdim}{M}		
\safemath{\sparsity}{S}	
\safemath{\inputdimA}{{N_a}}	
\safemath{\inputdimB}{{N_b}}	
\safemath{\elemA}{{n_a}}	
\safemath{\elemB}{{n_b}}	
\safemath{\resA}{\matR_a}	
\safemath{\resB}{\matR_b}	
\safemath{\subD}{\matS} 
\safemath{\subA}{\matS_a} 
\safemath{\subB}{\matS_b} 
\safemath{\dicta}{\matA} 	
\safemath{\dictb}{\matB} 	
\safemath{\hollowS}{H}
\safemath{\hollowA}{H_a}
\safemath{\hollowB}{H_b}
\safemath{\cross}{Z}
\safemath{\coh}{\mu_d}			
\safemath{\coha}{\mu_a}			
\safemath{\cohb}{\mu_b}			
\safemath{\mubs}{\nu}	
\safemath{\cohm}{\mu_m} 
\safemath{\dictset}{\setD}	
\safemath{\dictsetp}{\dictset(\coh,\coha,\cohb)}	
\safemath{\dictsetgen}{\dictset_\text{gen}}
\safemath{\dictsetgenp}{\dictsetgen(\coh)}
\safemath{\dictsetonb}{\dictset_\text{onb}}
\safemath{\dictsetonbp}{\dictsetonb(\coh)}

\safemath{\leftside}{U}
\safemath{\rightsideA}{R_a}
\safemath{\rightsideB}{R_b}

\safemath{\indexS}{\setI_S} 

\safemath{\na}{n_a}			
\safemath{\nb}{n_b}			
\safemath{\coeffa}{p_i}	
\safemath{\coeffb}{q_j}	
\safemath{\seta}{\setP}		
\safemath{\setb}{\setQ}     
\safemath{\setw}{\setW}	
\safemath{\setz}{\setZ}	
\safemath{\cola}{\veca}		
\safemath{\colb}{\vecb}		
\safemath{\cold}{\vecd}		
\safemath{\inputvec}{\vecx} 	
\safemath{\error}{\vece}	
\safemath{\noiseout}{\vecz} 	
\safemath{\inputvecel}{x}
\safemath{\inputveca}{\vecx_a}
\safemath{\inputvecb}{\vecx_b}
\safemath{\outputvec}{\vecy}	
\safemath{\lambdamin}{\lambda_{\mathrm{min}}}

\safemath{\elltwo}{\ell_2}
\safemath{\ellone}{\ell_1}
\safemath{\ellzero}{\ell_0}
\safemath{\ellinf}{\ell_\infty}
\safemath{\ellinftilde}{\ell_{\widetilde\infty}}
\safemath{\licard}{Z(\coh,\coha,\cohb)}
\safemath{\xsol}{\hat{x}}
\safemath{\xbord}{x_b}		
\safemath{\xstat}{x_s}		
\safemath{\xstatLone}{\tilde{x}_s}
\safemath{\order}{\mathcal{O}} 
\safemath{\scales}{\Theta} 
\safemath{\ones}{\mathbf{1}} 
\safemath{\zeroes}{\mathbf{0}} 
\safemath{\thlone}{\kappa(\coh,\cohb)} 
\safemath{\constoneA}{\delta} 
\safemath{\constoneB}{\epsilon} 
\safemath{\nlarge}{L}				   
\safemath{\sumlarge}{S_\nlarge}
\safemath{\maxlarger}{P_\nlarge}	   
\safemath{\Pzero}{\textrm{P0}}	
\safemath{\Pone}{\textrm{P1}}
\safemath{\vecfir}{\vecw}			 
\safemath{\vecsec}{\vecz}
\safemath{\elvecfir}{w}              
\safemath{\elvecsec}{z}				 
\safemath{\nlargefir}{n}
\safemath{\normout}{\gamma}
\safemath{\auxfun}{h}
\safemath{\supp}{\textrm{supp}}

\safemath{\indexa}{\ell}
\safemath{\indexb}{r}
\safemath{\indexc}{i}
\safemath{\indexd}{j}

\safemath{\project}{P}

\usepackage{framed}

\IEEEoverridecommandlockouts
\allowdisplaybreaks 

\definecolor{gray}{HTML}{555555}   

\newcommand*\tinygraycircled[1]{\Circled[inner color=white, fill color=mylightgray1, outer color=mylightgray1]{\textnormal{\footnotesize{#1}}}}

\safemath{\Hj}{\bJ}
\safemath{\bsj}{\bmw}
\safemath{\sj}{w}
\safemath{\Ej}{E_w}
\safemath{\proxg}{\text{prox}_g}
\safemath{\rE}{\rho_{\textsf{\tiny{E}}}}
\safemath{\rP}{\rho_{\textsf{\tiny{P}}}}

\renewcommand{\bSigma}{\mathbf{\Sigma}}

\begin{document}
\bstctlcite{IEEEexample:BSTcontrol} 

\title{Joint Jammer Mitigation and Data Detection}

\author{\IEEEauthorblockN{Gian Marti and Christoph Studer}
\thanks{A short version of this paper has been presented at IEEE ICC 2023 \cite{marti2023jmd}.}
\thanks{The work of GM and CS was supported in part by an ETH Research~Grant.}
\thanks{GM and CS are with the Department of Information Technology and Electrical Engineering, ETH Zurich, Switzerland; e-mail: gimarti@ethz.ch and  studer@ethz.ch}
}

\maketitle

\begin{abstract}
Multi-antenna (or MIMO) processing is a promising solution to the problem of jammer mitigation. Existing methods mitigate the jammer based on an estimate of its spatial signature that is acquired through a dedicated training phase. This strategy has two main drawbacks: (i) it reduces the communication rate since no data can be transmitted during the training phase and (ii) it can be evaded by smart or multi-antenna jammers that do not transmit during the training phase or that dynamically change their subspace through time-varying beamforming. 
To address these drawbacks, we propose \emph{Joint jammer Mitigation and data Detection (JMD)}, a novel paradigm for MIMO jammer mitigation. The core idea of JMD is to estimate and remove the jammer interference subspace \emph{jointly} with detecting the legitimate transmit data over multiple time slots. Doing so removes the need for a dedicated and rate-reducing training period while being able to mitigate smart and dynamic multi-antenna jammers.
We provide two JMD-type algorithms, SANDMAN and MAED, that differ in the way they estimate the channels of the legitimate transmitters
and achieve different complexity-performance tradeoffs. Extensive simulations demonstrate the efficacy of JMD for jammer mitigation. 

\end{abstract}


\section{Introduction}

Averting the threat of jamming attacks on wireless communication systems is
a problem of vital importance as society becomes ever more reliant on wireless communication infrastructure~\cite{threatvectors2021cisa, spacethreat2022, idr2022unmanned, economist2021satellite}. 
Multi-antenna (or MIMO) processing offers an attractive and effective solution by enabling the mitigation 
of jammers through spatial filtering \cite{pirayesh2022jamming}.
In MIMO-based jammer mitigation, the jammer's
interference is often removed by projecting the receive signal onto the orthogonal complement of the jammer channel's subspace, 
which traditionally is estimated in a dedicated training period 
\cite{shen14a, marti2021snips, yan2016jamming, do18a, akhlaghpasand20a,nguyen2021anti, leost2012interference, jiang2021efficient}.
Such an approach has two main disadvantages: First, a dedicated training period for estimating the jammer's channel reduces the achievable data rate, since no information can be transmitted during the training period. 
Second, a training-period-based approach is ineffective against smart jammers that transmit only at specific 
instances to evade estimation \cite{marti2023maed} or against dynamic jammers that change their subspace through 
time-varying multi-antenna transmit beamforming \cite{hoang2021suppression}.
In order to avoid both of these limitations, we propose a novel paradigm for jammer mitigation with MIMO processing that we call \emph{joint jammer mitigation and data detection~(JMD)}.

\subsection{State of the Art}
The fundamental challenge of jammer mitigation through spatial filtering is that
it requires information about the jammer, 
such as the subspace spanned by the jammer's channel \cite{yan2016jamming, shen14a, marti2021snips, do18a} 
or the covariance matrix of the jammer's interference \cite{marti2021snips, zeng2017enabling, yang2022estimation}.
Existing methods often assume that the jammer transmits permanently and with static beamforming.
Such a behavior would enable one to estimate the required quantities during a dedicated training period in which the legitimate transmitters do not transmit 
\cite{shen14a, marti2021snips, yang2022estimation} or in which they transmit predetermined symbols (pilots) that carry no information~\cite{yan2016jamming,do18a, akhlaghpasand20a,nguyen2021anti,leost2012interference}. 
Relying on such a  static jammer assumption,
the receiver can filter the jammer in the subsequent communication period until the wireless channel changes and 
the jammer training process is started anew.
A smart jammer, however, can easily circumvent such mitigation methods by deliberately violating 
their assumptions: It can pause jamming for the duration of the dedicated training period, 
so that the receiver learns nothing meaningful~\cite{marti2023maed}.
Or, if the jammer has multiple antennas, it can use beamforming to dynamically change its subspace
(as well as the interference covariance matrix at the receiver), so that after the training period, 
the receiver's filter will no longer match the jammer's transmit characteristics 
\cite{hoang2021suppression, hoang2022multiple}.\footnote{Another hard-to-mitigate threat is posed by multiple 
single-antenna distributed jammers, 
which also cause high-rank interference \cite{gulgun2020massive, vinogradova16a}.}

In order to mitigate smart jammers, methods have been suggested that attempt to fool the jammer
into transmitting during the training period by distributing and randomizing the timing of the training 
period \cite{shen14a, yan2016jamming}. However, such methods may not work against jammers that jam only 
intermittently at random time instants.
Similarly, methods have been proposed to mitigate dynamic multi-antenna jammers by perpetually  
estimating their instantaneous subspace~\cite{hoang2021suppression}. However, such methods may be effective only against jammers that change 
their subspace in a sufficiently slow manner. 
Furthermore, all training-period based mitigation methods are subject to an inherent trade-off between the time that 
they dedicate to the attempt of estimating the required jammer characteristics 
and the time that remains for payload data transmission. 

In light of these considerations, a more principled 
approach to MIMO-based jammer mitigation is needed. 
We have recently proposed MAED \cite{marti2023maed}, which, in hindsight, can be viewed as 
the first instance of our JMD paradigm. 
MAED unifies not only jammer mitigation and data detection, but also channel estimation (at the cost of increased 
computational complexity).

\subsection{Contributions}

We propose \emph{joint jammer mitigation and data detection (JMD)}, a novel paradigm for jammer mitigation. 
The core~idea is to estimate and remove the subspace of the jammer interference \emph{jointly}
with detecting the data of an entire transmission frame (or coherence interval). 
JMD removes the need for a dedicated training period, which enables higher data rates because more time is available
for data transmission. Beyond that, considering an entire transmission 
frame at once enables JMD to deal with smart jammers that try to evade mitigation 
(i)~by jamming only at specific instances or (ii) by dynamically changing their subspace through multi-antenna beamforming.

As in \cite{marti2023maed}, we leverage the fact that a jammer cannot leave its subspace within a 
coherence interval (which holds also for multi-antenna jammers with time-varying beamforming).
Going beyond our work~in~\cite{marti2023maed}, we show that this fact can 
be exploited to mitigate the jammer's contamination of the channel estimate even when the channel is 
estimated independently of the subsequent jammer mitigation and data detection. This insight opens the door 
to computationally more efficient signal-processing algorithms for jammer mitigation.\footnote{In particular, 
this insight has enabled the first application-specific integrated circuit (ASIC) implementation of an 
algorithm that mitigates smart jammers in multi-user MIMO \cite{bucheli2024jammer}.}
Moreover, our earlier work in~\cite{marti2023maed} only mitigates single-antenna jammers, 
while this paper shows that JMD is also well suited for the mitigation of multi-antenna jammers---even dynamic ones.
We capitalize on these new insights by proposing the JMD-type algorithm SANDMAN 
(an algorithm that offers all of the jammer mitigation 
capabilities of MAED \cite{marti2023maed}, but at reduced computational complexity), 
and by extending MAED to multi-antenna jammers. 
Furthermore, we present theoretical guarantees for JMD that are stronger than those in \cite{marti2023maed} 
\emph{and} that rely on weaker assumptions. 
Extensive simulations demonstrate the efficacy of our JMD-type algorithms SANDMAN and MAED for mitigating a wide 
range of smart and dynamic single- and multi-antenna jammers.
This paper therefore provides a full and mature account of the JMD paradigm, while the earlier work in \cite{marti2023maed} 
should (in hindsight) be regarded as a precursor that did not yet grasp the full breadth
and depth of the JMD paradigm. Compared to the conference version \cite{marti2023jmd}, this paper adds most of the theoretical 
results (and~\emph{all} of the proofs), more comprehensive discussion, and more extensive evaluation.

\subsection{Notation}
Matrices and column vectors are represented by boldface uppercase and lowercase letters, respectively.
For a matrix~$\bA$, the transpose is $\tp{\bA}$, the conjugate transpose is $\herm{\bA}$, 
the Moore-Penrose pseudoinverse is $\pinv{\bA}$ 
(if $\bA$ is full-rank and tall, then $\pinv{\bA}=\inv{(\herm{\bA}\bA)}\herm{\bA}$, 
and if $\bA$ is full-rank and wide, then $\pinv{\bA}=\herm{\bA}\inv{(\bA\herm{\bA})}$), 
and the Frobenius norm is $\| \bA \|_F$.
The columnspace and rowspace of $\bA$ are $\textit{col}(\bA)$ and $\textit{row}(\bA)$, respectively, 
and their orthogonal complements are $\textit{col}(\bA)^\orth$ and $\text{row}(\bA)^\orth$.
Horizontal concatenation of two matrices $\bA$ and~$\bB$ is denoted by $[\bA,\bB]$; vertical concatenation by $[\bA;\bB]$.
The $N\!\times\!N$ identity matrix is $\bI_N$.
The $\ell_2$-norm of a vector $\bma$ is $\|\bma\|_2$, 
and the $\ell_0$-``norm'' $\|\bma\|_0$ indicates the number of nonzero entries of $\bma$. 
We use $\setC\setN(\mathbf{0},\bC)$ to denote the $N$-dimensional circularly-symmetric 
complex Gaussian distribution with covariance matrix $\bC\in\opC^{N\times N}$. 
Finally, $[1:N]$ denotes the set of integers from $1$ through $N$.

\section{Transmission Model} \label{sec:model}

We focus on mitigating jamming attacks in the flat-fading, block-fading multi-user (MU) MIMO uplink.\footnote{
Our methods are also translatable to other MIMO contexts. For example, an extension to MIMO-OFDM in frequency-selective 
channels is possible but not straightforward, see \cite{marti2023single}.}
We consider the following transmission model:
\begin{align}
	\bmy_k = \bH \bms_k + \Hj \bsj_k + \bmn_k.  \label{eq:model}
\end{align}
Here, $\bmy_k\in\opC^B$ is the BS receive vector at sample instant~$k$,
$\bH \in \opC^{B\times U}$ is the UE channel matrix, 
$\bms_k\in\setS^U$ contains the sample-$k$ transmit symbols of $U$ single-antenna UEs
with constellation $\setS$ and is assumed to be uniformly distributed over $\setS^U$, 
$\bJ \in\opC^{B\times I}$ is the channel matrix of an \mbox{$I$-antenna} jammer,\footnote{The
model in \eqref{eq:model} can also represent distributed 
single- or multi-antenna jammers with a total of $I$ antennas. We consider this case in \fref{sec:smart}.} 
$\bmw_k\in\opC^{I}$ is the sample-$k$ jammer transmit vector, and \mbox{$\bmn_k\sim\setC\setN(\mathbf{0},N_0\bI_B)$}
is circularly-symmetric complex 
Gaussian noise with per-entry variance $N_0$.
We assume that $B\geq U+I$, i.e., the number of receive antennas is greater than or equal to the number of 
antennas of the UEs and the jammer combined, and that the matrix $[\bH,\bJ]$ has full rank $U+I$.
For simplicity, 
we take~$\setS$ to be QPSK, though other constellations would also be possible \cite{marti2023maed}. 
The constellation $\setS$ is scaled to unit symbol power, so that $\Ex{}{\bms_k\herm{\bms_k}}=\bI_U$, and
factors related to power control are absorbed into the channel matrix $\bH$.
In general, the jammer is a dynamic multi-antenna jammer and is able to dynamically change
its jamming activity.
Specifically, the jammer transmits
\begin{align} 
	\bsj_k = \bA_k \tilde\bsj_k, \label{eq:jammer_beamforming}
\end{align}
where, without loss of generality, $\Ex{}{\tilde\bmw_k\herm{\tilde\bmw_k}}=\bI_I$ for all $k$,
and $\bA_k\in\opC^{I\times I}$ is a beamforming matrix that can change arbitrarily over time 
(i.e., $\bA_k$ depends on the sample instant~$k$).
In particular, $\bA_k$ can be the all-zero matrix (no~jamming at sample instant~$k$), 
or some of its rows can be zero (the jammer uses only a subset of its antennas at sample instant $k$), 
or it can be rank-deficient in some other way. 

In JMD, the receiver will process the receive signal in blocks consisting of $K$ sample indices (or channel uses). 
The length~$K$ of these
blocks, which we call \emph{frames}, may be up to the length of a coherence interval, so that the channel 
matrices $\bH$ and $\bJ$ can be assumed to be constant for the duration of a frame. 
It will be useful to define notation for the receive signal of an entire frame:
\begin{align}
	\bY &= \bH\bS + \bJ\bW + \bN. \label{eq:matrix_io}
\intertext{
Here, $\bY=[\bmy_1,\dots,\bmy_K]\in\opC^{B\times K}$ is the receive matrix, and the quantities $\bS\in\opC^{U\times K}$, 
$\bW\in\opC^{I\times K}$, and $\bN\in\opC^{B\times K}$ are defined analogously.
We divide each frame into a \emph{pilot phase} of length $T\geq U$ 
and a \emph{data phase} of length $D$ (such that $T+D=K$). The input-output relations
for these phases are denoted by 
}
	\bY_T &= \bH\bS_T + \bJ\bW_T + \bN_T \label{eq:pilot_io}	
\intertext{
and
}
	\bY\!_D &= \bH\bS_D + \bJ\bW\!_D + \bN_D, \label{eq:data_io}	
\end{align}
respectively. 
Thus, we have $\bY=[\bY_T,\bY\!_D]$, $\bS=[\bS_T,\bS_D]$ with $\bS_T\in\opC^{U\times T}$ 
and $\bS_D\in\setS^{U\times D}$, $\bW=[\bW_T,\bW\!_D]$, and $\bN=[\bN_T,\bN_D]$.
Note that we do not restrict the pilots to the constellation $\setS$. However, we \emph{do} impose an average power constraint 
$\frac{1}{T}\|\bms_T\|_2^2=1$ on all rows $\bms_T$ of $\bS_T$. 
We assume that the columns $\bms_k$ of $\bS_D$ 
are independent of each other, so that $\bS_D\sim\text{Unif}[\setS^{U\times D}]$ (where $\text{Unif}[\setA]$ denotes 
the uniform distribution over the set $\setA$).
In what follows, we use plain letters for the true transmit signals and channels, 
letters with a tilde for optimization variables, and letters with a hat for (approximate) 
solutions to optimization problems. 
For example, $\hat\bS_D$ is an estimate of the data matrix $\bS_D$ obtained by 
solving an optimization problem with respect to a variable~$\tilde\bS_D$.

\section{Joint Jammer Mitigation and Data Detection} \label{sec:jmd}
Existing methods for MIMO jammer mitigation often null a jammer by 
projecting the receive signal onto the orthogonal complement of
the jammer subspace \cite{marti2021snips, yan2016jamming, do18a, hoang2021suppression}. 
That is, they compute $\bP\bmy_k$, where $\bP = \bI_B - \Hj\pinv\Hj$ is the orthogonal projection onto $\textit{col}(\Hj)^\orth$ 
and has the property $\bP\bJ=\mathbf{0}$.
After this projection, the transmit data can be detected using the effective channel matrix $\bH_\bP \triangleq \bP\bH$, 
since
\begin{align}
	\bP\bmy_k &= \bP\bH\bms_k + \bP\Hj\bA_k \tilde \bsj_k + \bP \bmn_k 
	= \bP\bH\bms_k + \bP \bmn_k \\
	&\triangleq \bH_\bP \bms_k + \bmn_{\bP,k}. \label{eq:projected_model}
\end{align}
This strategy mitigates the jammer regardless of which vectors~$\tilde\bsj_k$ and which beamforming matrices
$\bA_k$ the jammer uses (cf.~\eqref{eq:jammer_beamforming}), since $\textit{col}(\bJ\bA_k)\subseteq\textit{col}(\bJ)$ 
for any $\bA_k$. The problem, however, is that of estimating $\bJ$---or $\textit{col}(\bJ)$---when 
the jammer changes $\bA_k$ dynamically such that $\textit{col}(\bJ\bA_k)$ depends on $k$. 
In that case, the receiver will observe different interference subspaces $\textit{col}(\bJ\bA_k)$ at different instants, 
but possibly never the ``pure'' jammer subspace~$\textit{col}(\bJ)$.

To address this problem, we propose JMD.
The key idea is to consider jammer subspace estimation and nulling as well as 
data detection over an entire transmission frame (or coherence interval) \emph{jointly}. 
The jammer subspace is identified with the subspace that is not explainable in terms of UE transmit signals; 
simultaneously, the UE transmit data are estimated while projecting the receive signals onto the 
orthogonal complement of the jammer subspace estimate. 
Mathematically, this can be formulated as the following optimization problem:\footnote{For the moment, 
we neglect estimation of the UE channel matrix $\bH$ and assume perfect channel knowledge. 
Channel estimation is considered in \fref{sec:chest}. In this section, signals from the pilot phase \eqref{eq:pilot_io} are
therefore ignored.}
\begin{align}
	\min_{\substack{\tilde\bS_D \,\in\, \setS^{U\times D},\\ \tilde\bP \,\in\, \mathscr{G}_{B-I}(\opC^B)}} 
	\| \tilde\bP (\bY\!_D -\bH \tilde\bS_D) \|_F^2. \label{eq:jmd_problem}
\end{align}
Here, $\tilde\bS_D$ is the potential estimate of the data matrix $\bS_D$ 
and $\tilde\bP = \bI_B - \tilde\bJ\pinv{\tilde\bJ}$ is the projection onto the orthogonal
complement of the potential jammer subspace estimate $\textit{col}(\tilde\bJ)$, with $\tilde\bJ\in\opC^{B\times I}$.
The range of $\tilde\bP$ is the Grassmannian manifold 
$\mathscr{G}_{B-I}(\opC^B)=\{\bI_B-\tilde\bJ\pinv{\tilde\bJ}:\tilde\bJ\!\in\!\opC^{B\times I}\}$,
i.e., the set~of~orthogonal projections onto $(B\!-\!I)$-dimensional subspaces of~$\opC^B$.
The intuition behind this problem formulation is as follows: Using \eqref{eq:data_io}, 
we can rewrite the objective of \eqref{eq:jmd_problem} as
\begin{align}
	\|\tilde\bP( \bH(\bS_D-\tilde\bS_D) + \bJ\bW\!_D + \bN_D)\|_F^2. \label{eq:pre_intuition}
\end{align}
If we assume that the thermal noise $\bN_D$ is negligible \mbox{($\bN_D=\mathbf{0}$)} compared 
to the legitimate signals $\bH\bS_D$ and to the jammer interference $\bJ\bW\!_D$, 
then we may further simplify \eqref{eq:pre_intuition} to 
\begin{align}\textstyle
	\bigg\|\tilde\bP \begin{bmatrix}\bH,\bJ\end{bmatrix}
	\begin{bmatrix}\bS_D-\tilde\bS_D \\ \bW\!_D \end{bmatrix} \bigg\|_{F \,\displaystyle{.}}^2 \label{eq:intuition}
\end{align}
The justification for this low-noise assumption---which we make repeatedly throughout the paper---is  
that this paper's focus is \emph{jammer mitigation}, not data detection under high thermal noise.
In other words, we are concerned with scenarios where communication performance is not primarily limited by thermal noise, 
but by jamming. 
The objective in \eqref{eq:intuition} is minimized when the argument of the Frobenius norm is the all-zero matrix, 
and the projector $\tilde\bP$ can null a matrix of rank less than or equal to~$I$. 
If $\bJ\bW\!_D$ has rank $I$, this necessitates that $\tilde\bS_D=\bS_D$
(by assumption, the matrix $[\bH,\bJ]$ has full rank $U+I$; see \fref{sec:model}).
To make this intuition more rigorous, we introduce the notion of \emph{eclipsing}, which describes
a certain relation between the jammer transmit matrix $\bW\!_D$ and the error matrix $\bS_D-\tilde\bS_D$ 
corresponding to a potential symbol estimate~$\tilde\bS_D$. 
In this section, we define eclipsing while assuming perfect channel state information (CSI).
The more realistic case, in which CSI is estimated using pilots, is considered in~\fref{sec:chest}. 

\begin{defi}[Eclipsing with perfect CSI] \label{def:csi_eclipsing}
\mbox{The jammer} is eclipsed in a frame if there exists a matrix 
\mbox{$\tilde\bS_D\in\setS^{U\times D}\!\setminus\!\{\bS_D\}$} so that the 
virtual interference matrix $\bSigma(\tilde\bS_D)\triangleq[\bE(\tilde\bS_D); \bW\!_D]$ 
is a matrix of rank less than or equal to $I$, 
where $\bE(\tilde\bS_D) \triangleq \bS_D-\tilde\bS_D$ is the error matrix. 
\end{defi}

To say that a jammer is eclipsed therefore means that there exists an incorrect (i.e., distinct from $\bS_D$) 
estimate $\tilde\bS_D$ of the data matrix $\bS_D$ such that the virtual interference matrix $\mathbf{\Sigma}(\tilde\bS_D)$,
which appears in \eqref{eq:intuition}, has rank less than or equal to~$I$.
Our first theoretical result is that, under the assumption of negligible thermal noise, eclipsing is the only aspect 
that might hinder successful jammer nulling and data recovery:

\begin{thm}\label{thm:perfect_csi:correct}
	If the thermal noise is zero \textnormal{($\bN_D=\boldsymbol{0}$)}, and if the jammer is not eclipsed, 
	then the problem in \eqref{eq:jmd_problem} has the unique solution 
	$\{\hat\bP,\hat\bS_D\} = \{\bI_B - \bJ\pinv{\bJ}, \bS_D\}$.
\end{thm}

(All proofs are in the appendix.)
The question is therefore: What are the chances that a jammer is---or is not---eclipsed?
This question turns out to be, in general, quite complex.  
For instance, an immediate consequence of \fref{def:csi_eclipsing} is that any jammer with $\textit{rank}(\bW\!_D)<I$
is eclipsed. The ``problem'' with such a jammer that only jams with rank smaller than $I$ is that a projection 
$\tilde\bP\in\mathscr{G}_{B-I}(\opC^B)$ can null more dimensions than are occupied by interference, 
and so can also null any data detection errors $\bS_D-\tilde\bS_D$ of rank up to $I-\textit{rank}(\bW\!_D)$, 
which jeopardizes accurate data detection. In that sense, it would therefore be preferable to identify 
the parameter $I$ not with the number of physical jammer antennas $I$, but with the effective rank of the jammer interference, 
$\textit{rank}(\bW\!_D)$. (Any receive interference $\bJ\bW\!_D$ with $\textit{rank}(\bW\!_D)<I$ can be rewritten 
also as a product $\bJ'\bW_{\!D}'$ for some $\bJ'\in\opC^{B\times\textit{rank}(\bW\!_D)}$.)
Nevertheless, we have decided to identify $I$ with the number of jammer antennas because this is more natural from an operational point
of view. We note, however, that our working assumption will therefore be that the jammer transmits with full antenna rank, 
i.e., \mbox{$\textit{rank}(\bW\!_D)=I$.}
As a consequence of this difficulty, we can only give strong theoretical guarantees for the probability of eclipsing for single-antenna 
jammers.\footnote{Also, it is evident from our proofs that the combinatorics involved in the analysis are 
quite involved already for single-antenna jammers, and we expect them to significantly increase in difficulty for multi-antenna~jammers.}
But we emphasize that our empirical results demonstrate that JMD works well also for multi-antenna jammers; see \fref{sec:results}.
For single-antenna jammers, for which we write $\bJ\bW\!_D$ as $\bmj\tp{\bmw_D}$, 
the probability of eclipsing turns out to depend on the number of symbols that the jammer jams 
\mbox{(i.e., on $\|\bmw_D\|_0$).}
A jammer that never jams is guaranteed to eclipse; a jammer that jams perpetually eclipses only with
minuscule probability:

\begin{thm} \label{thm:csi_eclipsing}
	The probability that a single-antenna jammer with transmit signal~$\bmw_D$ eclipses is upper bounded by
	\vspace{-1mm}
	\begin{align}
		p_e(\bmw_D) = 
		\begin{cases}
			1 &\text{if } \bmw_D = \mathbf{0} \\
			1-\left(1-\frac{2^{\|\bmw_D\|_0}-1}{4^{\|\bmw_D\|_0-1}}\right)^U & \text{else.}
		\end{cases}
\label{eq:csi_eclipsing} 
	\end{align}
	If $\|\bmw_D\|_0\gg0$, then $p_e(\bmw_D) \approx \tilde p_e(\bmw_D)\triangleq4U\cdot2^{-\|\bmw_D\|_0}$.
\end{thm}

\fref{fig:collision_probability} depicts this probability as a function of the number of jammed symbols $\|\bmw_D\|_0$
for different numbers of UEs~$U$.

\section{Channel Estimation} \label{sec:chest}
In the previous section, we have assumed that the BS knows the UE channel matrix $\bH$ perfectly.
In practice, however, $\bH$ has to be estimated. This is typically done using pilots. 
In the presence of jamming, the obtained estimate can be contaminated
\cite{clancy2011efficient, sodagari2012efficient, shahriar2012performance, nguyen2023pilot}. 
We now propose two different approaches for jammer-resilient channel estimation in combination with~JMD.

\subsection{Linear Channel Estimation} \label{sec:indep}

The first approach consists of simply using a classical linear channel 
estimator such as least squares (LS) or linear minimum mean square error (LMMSE). 
It may not be immediately apparent why these estimates should be jammer-resilient. 
In fact, there is nothing inherently jammer-resilient about them; they become jammer-resilient only when 
combined with JMD.

The underlying mechanism works as follows: Let the pilot phase be given as in~\eqref{eq:pilot_io}, with 
a square ($T=U$) or wide ($T>U$) pilot matrix, and let the corresponding data phase be 
given as in~\eqref{eq:data_io}. 
If a linear channel estimator is used, then the channel estimation error due to the jammer interference is 
contained in $\textit{col}(\bJ)$. For instance, in the case of LS channel estimation, we have
\begin{align}
	\hat\bH &= \bY_T \pinv{\bS_T} 
	= \bH + \underbrace{\bJ \bW_T\pinv{\bS_T}}_{\in\,\textit{col}(\bJ)\hspace{-7.5mm}} + \bN_T\pinv{\bS_T}. \label{eq:pilot_contamination}
\end{align}
Therefore, if we simply plug the jammer-contaminated linear channel estimate $\hat\bH$ from \eqref{eq:pilot_contamination}
into \eqref{eq:jmd_problem} as follows,
\begin{align}
	\min_{\substack{\tilde\bS_D \,\in\, \setS^{U\times D},\\ \tilde\bP \,\in\, \mathscr{G}_{B-I}(\opC^B)}} 
	\| \tilde\bP (\bY\!_D - \hat\bH \tilde\bS_D) \|_F^2, \label{eq:jmd_problem2}
\end{align}
then the ``true'' jammer-nulling projection $\bP = \bI_B - \Hj\pinv\Hj$ also removes the jammer contamination of 
the channel estimate, since $\bP\hat\bH=\bP\bH$.
This implies that JMD can naturally incorporate jammer-resilient channel estimation, without requiring
extra countermeasures to protect channel estimation from jammer interference.

\begin{figure}[tp]
\centering
\includegraphics[height=4cm]{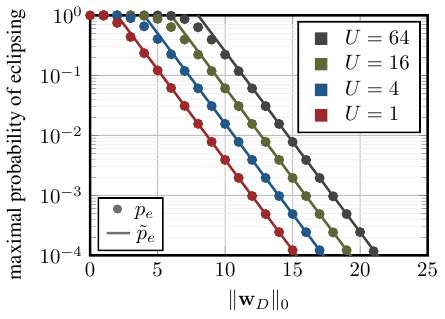}
\vspace{-3mm}
\caption{The bound $p_e$ (and its approximation $\tilde p_e$) on the probability that a single-antenna jammer eclipses vs.
the number of jammed symbols $\|\bmw\|_0$, for different numbers $U$ of~UEs.
The bound on the probability of eclipsing decreases exponentially in the number of jammed symbols. 
}
\vspace{-4mm}
\label{fig:collision_probability}
\end{figure}

To characterize the probability of successful data recovery for JMD in combination with linear channel estimation, 
we provide the following adapted (compared to \fref{def:csi_eclipsing}) definition of eclipsing:\footnote{
This definition of eclipsing coincides with \cite[Def.\,1]{marti2023maed} for $I\!=\!1$ and so can be seen 
as a generalization of the concept of eclipsing to multi-antenna~jammers.}

\begin{defi}[Eclipsing with channel estimation] \label{def:eclipse_chest}
The jammer is eclipsed in a given frame if there exists a matrix 
$\tilde\bS_D\in\setS^{U\times D}\setminus\!\{\bS_D\}$, so that
\mbox{$\mathbf{\Sigma}(\tilde\bS_D)\triangleq[\bS_D-\tilde\bS_D;$} $\bW\!_D - \bW_T\pinv{\bS_T}\tilde\bS_D]$ is a matrix of rank less than or equal to $I$. 
\end{defi}

With this modified definition of eclipsing, we obtain the following guarantee 
for JMD with linear channel estimation:

\begin{thm} \label{thm:independent_chest:correct}
	If the thermal noise is zero \textnormal{($\bN=\boldsymbol{0}$)}, and if the jammer is not eclipsed, 
	then the problem in \eqref{eq:jmd_problem2} has the unique solution 
	$\{\hat\bP,\hat\bS_D\} = \{\bI_B - \bJ\pinv{\bJ}, \bS_D\}$.
\end{thm}

As in \fref{sec:jmd}, the question becomes: what is the probability that 
a jammer eclipses?\footnote{As in \fref{sec:jmd}, we can characterize this probability analytically for 
the single-antenna jammer case, 
and we refer to the simulations of \fref{sec:results} for showing the efficacy against multi-antenna jammers empirically.}
The answer depends fundamentally on whether or not the jammer knows the pilots
that are being used by the UEs. Let us first consider the case where the jammer knows 
the pilots (as well as the transmit constellation $\setS$).

\begin{thm} \label{thm:independent_chest_eclipsing_nopilot}
If a single-antenna jammer knows at least one pilot sequence (i.e., one row of $\bS_T$)
as well as the constellation~$\setS$, then it can jam all frame symbols 
(i.e., $\|\bmw_T\|_0=T$ and $\|\bmw_D\|_0=D$) while eclipsing with probability $1-4^{-D}\approx1$. 
\end{thm}

However, the situation is completely different if the (single-antenna) jammer does \emph{not} know the pilots. 
This can be modeled by assuming that a (square) pilot matrix $\bS_T$ is 
drawn at random according to the power-normalized Haar measure 
(i.e., drawn uniformly from the set of unitary $T\times T$ matrices and scaled to $\|\bS_T\|_F^2=U^2$)
\cite{meckes2019random} and not revealed to the jammer in advance:\footnote{Note that the following result is much stronger
than \cite[Thm.\,2]{marti2023maed}, which assumes that both $\bmw_D$ and $\bmw_T$ are distinct from zero, 
and which gives a nontrivial bound only if the number of UEs satisfies $U>3$.}

\begin{thm} \label{thm:independent_chest_eclipsing}
If $\bS_T\in\opC^{T\times U}$ is square ($T=U$) and drawn according to the power-normalized Haar measure 
without being revealed to the jammer, then the probability that a single-antenna jammer eclipses is upper bounded by 
	\begin{align}
		p_e(\bmw_D) = 
		\begin{cases}
			1 &\text{if } \bmw_D = \mathbf{0} \\
			1-\left(1-\frac{2^{\|\bmw_D\|_0}-1}{4^{\|\bmw_D\|_0-1}}\right)^U & \text{else,}
		\end{cases}
	\label{eq:ind_chest_eclipsing}
	\end{align}
if the jammer does not jam the pilot phase (i.e., if $\bmw_T=\mathbf{0}$), and otherwise (i.e., if $\bmw_T\neq\mathbf{0}$),
it is upper bounded by
\begin{align}
	p_{e,\bmw_T\neq\mathbf{0}} = 1-\left(1-\frac{2^D-1}{4^D-1}\right)^U. \label{eq:jammer_jams_pilots}
\end{align}
\end{thm}

Note that the bound in \eqref{eq:ind_chest_eclipsing}, where the jammer is silent during the pilot phase,
coincides with the bound from \fref{thm:csi_eclipsing}. 
The bound in \eqref{eq:jammer_jams_pilots}, where the jammer is active during the pilot 
phase, satisfies $p_{e,\bmw_T\neq\mathbf{0}}\leq p_e(\bmw_D)$ with equality if and only if $\bmw_D$ 
has full support, in which case $p_{e,\bmw_T\neq\mathbf{0}}= p_e(\bmw_D) \approx 4U\cdot2^{-D}$.

It is evident that the guarantees for successful JMD jammer mitigation are fundamentally dependent on 
whether or not the jammer knows the pilots used for channel estimation. 
This finding is reminiscent of a more general information-theoretic result which shows that reliable communication 
in the presence of jamming is possible if the legitimate transmitter and the receiver share 
a common secret, but not otherwise\mbox{\cite[Sec.\,V]{lapidoth1998reliable}.}
Intuitively, the receiver needs a lever that enables it to distinguish between legitimate
and illegitimate transmit signals. This lever could be knowledge of the legitimate UEs' channel 
matrix $\bH$ (as in \fref{sec:jmd}). Or, if ground-truth channel knowledge is not available, 
it could be a mechanism which guarantees that the jammer cannot replicate a legitimate pilot 
sequence during the channel estimation phase (as in \fref{thm:independent_chest_eclipsing}).

\subsection{Joint Channel Estimation} \label{sec:joint}

Besides linear channel estimation, there is a second possible approach for channel estimation in combination with JMD,
namely to estimate the channel \emph{jointly} with the data and the 
jammer-mitigating projection.
Specifically, the idea is to use a pilot and a data phase as in \eqref{eq:pilot_io} and~\eqref{eq:data_io}, 
and to then solve
\begin{align}
	\min_{\substack{
	\tilde\bP \,\in\, \mathscr{G}_{B-I}(\opC^B),\\
	\tilde\bH \,\in\, \opC^{B\times U},\\ 
	\tilde\bS_D \,\in\, \setS^{U\times D}
	}} 
	\big\| \tilde\bP ([\bY_T, \bY\!_D] - \tilde\bH [\bS_T, \tilde\bS_D]) \big\|_F^2. \label{eq:jmd_problem3}
\end{align}
Analogous to linear channel estimation (cf. \fref{sec:indep}), 
this approach exploits the fact that the jammer's interference during 
the channel estimation and the data phase is contained in the same subspace $\textit{col}(\bJ)$.
Compared to linear channel estimation, joint channel 
estimation benefits from a joint channel estimation and data detection (JED) performance gain, 
at the price of having to solve a more complex optimization problem. 

To characterize the probability of successful data recovery for JMD in combination with joint channel 
estimation, we can use the same definition of eclipsing (\fref{def:eclipse_chest}) as in \fref{sec:indep}.
Furthermore, it turns out that we also obtain the same success guarantee:

\begin{thm} \label{thm:joint_chest:correct}
	If the noise is zero \textnormal{($\bN=\boldsymbol{0}$)}, and if the jammer is not eclipsed, 
	then the problem in \eqref{eq:jmd_problem3} has the unique solution 
	$\{\hat\bP,\hat\bP\hat\bH,\hat\bS_D\} = \{\bI_B - \bJ\pinv{\bJ}, \hat\bP\bH, \bS_D\}$.\footnote{Because
	of the projection $\tilde\bP$ in \eqref{eq:jmd_problem3}, the optimal value $\hat\bH$ of $\tilde\bH$ itself is not
	uniquely determined---only its composition with the projection $\hat\bP$ is.	}
\end{thm}

Since the definition for eclipsing is the one from \fref{def:eclipse_chest},
the results of \fref{thm:independent_chest_eclipsing_nopilot} and \fref{thm:independent_chest_eclipsing} 
specify the probability of eclipsing also for the case of joint channel estimation.
Therefore, regardless of whether we estimate the channel matrix separately with a linear estimator 
(as in \fref{sec:indep}) or jointly (as in this section), randomized pilots are necessary for meaningful success guarantees. 
However, we emphasize that these guarantees do not imply that joint channel estimation has the same performance as linear
channel estimation under non-zero noise, i.e., for $\bN\neq\mathbf{0}$.

\section{Algorithms}
In the previous two sections, we have introduced the JMD paradigm (see \fref{sec:jmd}) and outlined two 
alternatives for channel estimation (linear and joint; see \fref{sec:chest}).
In this section, we provide concrete algorithms, i.e., we propose algorithms for approximately 
solving the optimization problems in \eqref{eq:jmd_problem2} and~\eqref{eq:jmd_problem3}.
The presented algorithms are by no means the only possible algorithms for these optimization problems. 
The development of other JMD algorithms is left as future~work.

\subsection{The SANDMAN Algorithm}
We start by providing SANDMAN, an algorithm for JMD in conjunction with linear channel estimation as expressed
through the optimization problem in \eqref{eq:jmd_problem2}.
Since solving \eqref{eq:jmd_problem2} exactly is difficult, we solve it approximately.
A key difficulty is that, due to the discreteness of $\setS$, the problem in \eqref{eq:jmd_problem2} is
NP-hard even when fixing $\tilde\bP$ and solving only for~$\tilde\bS_D$~\cite{grotschel2012geometric}.
We thus relax the constraint set $\setS$ to its convex hull $\setC\triangleq \textit{conv}(\setS)$.\footnote{
This relaxation, which allows us to use gradient-based methods for optimization, is only temporary---ultimately, 
entrywise rounding has to be used for converting a penultimate estimate $\tilde\bS_D\in\setC^{U\times D}$ 
to one that is contained in $\setS^{U\times D}$. For simplicity, our algorithms omit this detail.}

To promote symbol estimates at, or near, the corners of~$\setC$, 
we add a concave regularizer $-\|\tilde\bS_D\|_F^2$ weighted by $\alpha>0$ to the objective in \eqref{eq:jmd_problem2} \cite{shah2016biconvex}.
We colloquially refer to the resulting constraint and regularizer as a \emph{box prior}.
The use of this box prior is motivated by our assumption (cf. \fref{sec:model}) that the transmit constellation is QPSK,
so that the corners of $\setC$ correspond to the constellation symbols $\setS$.
If larger constellations such as 16-QAM or 64-QAM are used, other signal priors are more effective~\cite{marti2023maed}.
The relaxed problem is
\begin{align}
	\min_{\substack{\tilde\bS_D \,\in\, \setC^{U\times D},\\ \tilde\bP \,\in\, \mathscr{G}_{B-I}(\opC^B)}} 
	\big\| \tilde\bP (\bY\!_D - \hat\bH \tilde\bS_D) \big\|_F^2 - \alpha\|\tilde\bS_D\|_F^2. \label{eq:sandman_problem}	
\end{align}
Due to the nonconvex constraint set $\mathscr{G}_{B-I}(\opC^B)$, the relaxed problem is still non-convex.
But the following result~holds:

\begin{thm} \label{thm:convex}
	When $\tilde\bP$ is fixed and $\alpha\leq \lambda_{\min}(\hat\bH^{\textnormal{H}}\tilde\bP\hat\bH)$, 
	then the objective in \eqref{eq:sandman_problem} is convex in $\tilde\bS_D$.
	Vice versa, when $\tilde\bS_D$ is fixed, then the objective in \eqref{eq:sandman_problem} is minimized with 
	respect to $\tilde\bP$ by $\bI_B - \bU_I\herm{\bU_I}$, where $\bU_I\in\opC^{B\times I}$ consists of the $I$ 
	principal left-singular vectors of $\bY\!_D - \hat\bH \tilde\bS_D$ (i.e., the left-singular vectors corresponding
	to the $I$ largest singular values).
\end{thm}

This theorem suggests using an alternating minimization strategy, since solving \eqref{eq:sandman_problem}
for either $\tilde\bS_D$ or $\tilde\bP$ is straightforward as long as the other quantity is fixed:

\subsubsection{Solving for $\tilde\bS_D$}
To solve \eqref{eq:sandman_problem} with respect to $\tilde\bS_D$ (for fixed $\tilde\bP$), we use forward-backward splitting (FBS) \cite{goldstein16a}. 
FBS is a method for solving optimization problems of the~form
\begin{align}
	\min_{\tilde\bms}~ f(\tilde\bms) + g(\tilde\bms), \label{eq:fbs_problem}
\end{align}
where $f$ is convex and differentiable, and $g$ is convex but not necessarily differentiable. 
FBS solves such problems iteratively by computing 
\begin{align}
	\tilde\bms^{(t+1)} = \proxg\big(\tilde\bms^{(t)} - \tau^{(t)}\nabla f(\tilde\bms^{(t)}); \tau^{(t)}\big), \label{eq:fbs_iteration}
\end{align}
where $\tau^{(t)}$ is the stepsize at iteration $t$, $\nabla f(\tilde\bms)$ is the gradient~of $f$ in $\tilde\bms$,  
and $\proxg$ is the proximal operator of $g$, defined as \cite{parikh13a}
\begin{align}
\proxg(\tilde\bms; \tau) = \arg\!\min_{\tilde\bmx~~~~} \tau g(\tilde\bmx) + \frac12 \|\tilde\bms - \tilde\bmx\|_2^2. 
\end{align}
FBS solves convex problems exactly (for a sufficient number of iterations with suitable stepsizes~$\tau^{(t)}$), 
but it is also effective for approximately solving non-convex problems \cite{goldstein16a}.
To solve the problem in \eqref{eq:sandman_problem}, we define the functions $f$ and~$g$ as
\begin{align}
	f(\tilde\bS_D) &= \big\| \tilde\bP (\bY\!_D - \hat\bH \tilde\bS_D) \big\|_F^2, \\
	g(\tilde\bS_D) &= - \alpha\big\|\tilde\bS_D\big\|_F^2 + \chi_\setC(\tilde\bS_D), 
\end{align}
where $\chi_\setC$ acts entrywise on $\tilde\bS_D$ as the indicator function of $\setC$,
\begin{align}
	\chi_\setC(\tilde s) = \begin{cases}
		0 &: \tilde s \in \setC \\
		\infty &: \tilde s \notin \setC.
	\end{cases}
\end{align}
The gradient of $f$ in $\tilde\bS_D$ is given as 
\begin{align}
	\nabla f(\tilde\bS_D) = -2\,\herm{\hat\bH}\tilde\bP(\bY\!_D - \hat\bH \tilde\bS_D).
\end{align}
The proximal operator of $g$ acts entrywise on $\tilde\bS_D$ and is given~as
$\proxg(\tilde s; \tau) = \text{clip}(\tilde s/(1-\tau\alpha); \sqrt{\sfrac{1}{2}})$ when $\alpha\tau<1$ 
(where $\text{clip}(z;a)$ clips the real and imaginary part of $z\in\opC$ to the~interval $[-a,a]$),
and otherwise as 
$\arg\!\min_{\tilde x \in \{\pm\sqrt{\frac{1}{2}} \pm i\sqrt{\frac{1}{2}} \} } |\tilde s - \tilde x|^2$.

\begin{algorithm}[tp]
  \caption{Approximate SVD \cite{liberty2013svd}}
  \label{alg:svd}
  \begin{algorithmic}[1]
	\setstretch{1.0}
	\vspace{1mm}
    \Function{approxSVD}{$\mathbf{E},I$}
    \For{$i=1$ {\bfseries to} $I$}
    	\State draw $\bmx \sim \setC\setN(\mathbf{0},\bI_K)$
    	\State $\bmx' = \herm{\bE}\bE\bmx$
    	\State $\bmv_i = \bmx'/\|\bmx'\|_2$
    	\State $\sigma_i = \|\bE\bmv_i\|_2$
    	\State $\bmu_i = \bE\bmv_i/\sigma_i$
    	\State $\bE = \bE - \sigma_i\bmu_i\herm{\bmv_i}$
    \EndFor
    \State \textbf{output:} $[\bmu_1,\dots,\bmu_I]$
    \EndFunction
    
  \end{algorithmic}
\end{algorithm}

\subsubsection{Solving for $\tilde\bP$}
According to \fref{thm:convex}, we can solve~\eqref{eq:sandman_problem} with respect to $\tilde\bP$ (for fixed $\tilde\bS_D$) 
by calculating the $I$ principal left-singular vectors $[\bmu_1,\dots,\bmu_I]=\bU_I$ of $\bY\!_D - \hat\bH \tilde\bS_D$.
Instead of performing an exact but computationally expensive singular value decomposition (SVD), 
we approximate $\bU_I$ with the power method from \cite{liberty2013svd}.
We perform a single power iteration per dimension. The resulting procedure is outlined in~\fref{alg:svd}.
\vspace{2mm}

The SANDMAN algorithm alternates between descent steps in $\tilde\bS_D$ and approximate computations of $\tilde\bP$ 
for a total of $t_\text{max}$ iterations. We choose $\alpha=2.5$, and the stepsizes $\tau^{(t)}$
are selected using the Barzilai-Borwein criterion \cite{barzilai1988two} detailed in \cite{zhou2006gradient}.
SANDMAN is summarized in \fref{alg:sandman} and has a computational complexity of
$O(t_\text{max}UDB)$, i.e., its complexity is linear in the number of UEs $U$, the number of data samples~$D$, 
and the number of BS antennas $B$.

One key feature of SANDMAN is that, for single antenna jammers 
for which $\pinv{(\tilde\bJ^{(t)})}=\pinv{(\tilde\bmj^{(t)})} = \herm{(\bmj^{(t)})}/\|\bmj^{(t)}\|^2$,
it requires \emph{no} matrix inversion. This makes SANDMAN particularly attractive for hardware implementation. 
In fact, the first (and so far the only) jammer-mitigating MU-MIMO receiver ASIC~\cite{bucheli2024jammer} is based on SANDMAN, mainly 
thanks to the fact that no matrix inversion is required.

\floatstyle{spaceruled}
\restylefloat{algorithm}
\begin{algorithm}[tp]
  \caption{SANDMAN}
  \label{alg:sandman}
  \begin{algorithmic}[1]
	\setstretch{1.0}
	\vspace{1mm}
    \Function{SANDMAN}{$\bY\!_D, \bY_T, \bS_T, I, t_\text{max}$}
    \State  $\hat\bH = \bY_T \pinv{\bS_T}$ \hfill \textcolor{gray}{// LS channel estimate}
    \State $\tilde\bS_D^{(0)} = \mathbf{0}_{U\times D}$
    \State $\tilde\bE^{(t)} = [\bY_T,\bY\!_D]$
    \For{$t=0$ {\bfseries to} $t_\text{max}-1$}
		\State $\tilde\bJ^{(t)} = \textsc{approxSVD}(\tilde\bE^{(t)}, I)$ 
		\hfill\textcolor{gray}{// cf. Algorithm 1} 
		\State $\tilde\bP^{(t)} = \bI_B - \tilde\bJ^{(t)}\pinv{(\tilde\bJ^{(t)})}$
		\State $\nabla f(\tilde\bS_D^{(t)}) = -2\,\herm{\hat\bH}\tilde\bP^{(t)}(\bY\!_D - \hat\bH \tilde\bS_D^{(t)})$
		\State $\tilde\bS_D^{(t+1)} = \proxg\big(\tilde\bS_D^{(t)} - \tau^{(t)}\nabla f(\tilde\bS_D^{(t)}); \tau^{(t)}\big)$
		\State $\tilde\bE^{(t+1)} = [\bY_T,\bY\!_D] - \hat\bH [\bS_T,\tilde\bS_D^{(t+1)}]$
    \EndFor
    \State \textbf{output:} $\tilde\bS_D^{(t_\text{max})}$    
    \EndFunction    
  \end{algorithmic}
\end{algorithm}

\subsection{The MAED Algorithm}

MAED is an algorithm for JMD in conjunction with joint channel estimation as expressed through the optimization problem
in \eqref{eq:jmd_problem3}. A version of MAED for single-antenna jammers has originally been proposed in 
\cite{marti2023maed}.\footnote{Apart from the number of jammer antennas, the MAED version provided in~\cite{marti2023maed} 
differs slightly from the one provided in this paper. The conference version of this paper \cite{marti2023jmd}
has therefore referred to the newer version of MAED as ``MAED~2.0.'' 
For simplicity, however, we have decided to return to the original name ``MAED,'' 
since these minor differences between the algorithms do not seem to justify a new name.}
Estimating the channel jointly enables MAED to achieve the detection gains associated
with joint channel estimation and data detection (JED) 
\cite{vikalo2006efficient, xu2008exact, kofidis2017joint, yilmaz2019channel, he2020model}.
Other than estimating the channel jointly as well, however, MAED works similarly to SANDMAN. 

To derive the MAED algorithm, we start by noting that the objective in \eqref{eq:jmd_problem3} is quadratic in $\tilde\bH$. 
The optimal $\tilde\bH$ as a function of $\tilde\bP$ and $\tilde\bS\triangleq[\bS_T,\tilde\bS_D]$ is therefore equal to
\begin{align}
	\tilde\bH = \tilde\bP\bY\pinv{\tilde\bS}.
\end{align}
By plugging this expression back into \eqref{eq:jmd_problem3}, we obtain an optimization problem which only depends on 
$\tilde\bP$ and $\tilde\bS$. Furthermore, as in \eqref{eq:sandman_problem}, we also relax the constraint 
set $\setS$ to its convex hull $\setC$ and add a concave regularizer $-\alpha\|\tilde\bS_D\|_F^2$ to promote symbol 
estimates near the constellation points. The resulting optimization problem is therefore 
\begin{align}
	\min_{\substack{
	\tilde\bP \,\in\, \mathscr{G}_{B-I}(\opC^B),\\
	\tilde\bS=[\bS_T,\tilde\bS_D]:\, \tilde\bS_D \,\in\, \setC^{U\times D}
	}} 
	\big\| \tilde\bP \bY (\bI_K - \pinv{\tilde\bS}\tilde\bS) \big\|_F^2 - \alpha\|\tilde\bS_D\|_F^2. 	
\end{align}
As in SANDMAN, we now use an alternating minimization strategy where we alternate between 
FBS-based descent steps in $\tilde\bS$ and approximate optimization steps in $\tilde\bP$. 
The gradient required for FBS is 
\begin{align}
	\nabla f(\tilde\bS) = - \herm{(\bY\pinv{\tilde\bS})}\tilde\bP\bY(\bI_K-\pinv{\tilde\bS}\tilde\bS).
\end{align}
The proximal operator $\text{prox}_g$ maps the first $T$ columns of $\tilde\bS$ to $\bS_T$ 
and acts entrywise on the remaining columns as 
$\text{prox}_g(\tilde s; \tau) = \text{clip}(\tilde s/(1-\tau\alpha); \sqrt{\sfrac{1}{2}})$ when $\alpha\tau<1$ 
(where $\text{clip}(z;a)$ clips the real and imaginary part of $z\in\opC$ to the~interval $[-a,a]$),
and otherwise as 
$\arg\!\min_{\tilde x \in \{\pm\sqrt{\frac{1}{2}} \pm i\sqrt{\frac{1}{2}} \} } |\tilde s - \tilde x|^2$.

As in SANDMAN, we choose $\alpha=2.5$ and compute the stepsizes
using the Barzilai-Borwein criterion of \cite{zhou2006gradient}.
The resulting algorithm is summarized in \fref{alg:maed} and has computational complexity $O(t_\text{max}UKB)$. 
While this is the same asymptotic complexity order as for SANDMAN, MAED requires the computation 
of a pseudoinverse (of $\tilde\bS\in\opC^{U\times K}$) in every algorithm iteration, which makes it less 
attractive for hardware~implementation.

\begin{algorithm}[tp]
  \caption{MAED}
  \label{alg:maed}
  \begin{algorithmic}[1]
	\setstretch{1.0}
	\vspace{1mm}
    \Function{MAED}{$\bY_T, \bY\!_D, \bS_T, I, t_\text{max}$}
    \State $\tilde\bS^{(0)} \!=\! \left[\bS_T, \mathbf{0}_{U\!\times\! D}\right]$
	\State $\tilde\bE^{(0)} = [\bY_T,\bY\!_D]$
    \For{$t=0$ {\bfseries to} $t_\text{max}-1$}
    	\State $\tilde\bJ^{(t)} = \textsc{approxSVD}(\tilde\bE^{(t)}, I)$ 
		\hfill\textcolor{gray}{// cf. Algorithm 1} 
		\State $\tilde\bP^{(t)} = \bI_B - \tilde\bJ^{(t)}\pinv{(\tilde\bJ^{(t)})}$
		\State $\nabla f(\tilde\bS^{(t)}) = -\herm{\big(\bY\tilde\bS^{(t)}{}^\dagger\big)} \tilde\bP^{(t)}\bY(\bI_K - \tilde\bS^{(t)}{}^\dagger \tilde\bS^{(t)})$
		\State $\tilde\bS^{(t+1)} = \proxg\big(\tilde\bS^{(t)} - \tau^{(t)}\nabla f(\tilde\bS^{(t)}); \tau^{(t)}\big)$
		\State $\tilde\bE^{(t+1)} = [\bY_T,\bY\!_D](\bI_K - \tilde\bS^{(t+1)}{}^\dagger \tilde\bS^{(t+1)})$
    \EndFor
    \State \textbf{output:} $\tilde\bS^{(t_\text{max})}_{[T+1:K]}$
    
    \EndFunction
    
  \end{algorithmic}
\end{algorithm}

\section{Evaluation}
\label{sec:results}

\subsection{Simulation Setup} \label{sec:setup}
We evaluate our JMD-type methods SANDMAN and MAED using simulations. 
The channel vectors are generated with QuaDRiGa \cite{jaeckel2014quadriga}. 
We simulate a \mbox{MU-MIMO} system with $B=32$ BS antennas and $U=16$ single-antenna UEs
at a carrier frequency of 2\,GHz 
using the 3GPP 38.901 urban macrocellular (UMa) channel model~\cite{3gpp22a}. All antennas are omnidirectional.
The BS antennas are arranged as a uniform linear array (ULA) and spaced at half wavelength.
The UEs are uniformly distributed at distances between $10$\,m and $250$\,m in a $120^\circ$ angular sector 
in front of the BS, and with a minimum angular separation of $1^\circ$ between any two UEs. 
We assume $\pm3$\,dB per-UE power control.
The specific jammer model varies between the different experiments, see below. In general, we consider 
$J\geq1$ jammers placed randomly in the same area as the UEs, with a minimum angular separation 
of $1^\circ$ between any two jammers as well as between any jammer and any UE. 
Every jammer is equipped with $\frac{I}{J}\geq1$ antennas arranged as a ULA with half-wavelength 
spacing that faces towards the~BS. 

To be able to compare our methods with training-period-based methods, we now consider communication frames 
of length $L=K+R$ instead of $K$. That is, we replace \eqref{eq:matrix_io} with 
\begin{align}
	\bY = \bH\bX + \bJ\bW + \bN \in \opC^{B\times L}, 
\end{align}
where the columns of $\bX$ consist of $R$ (for ``redundancy'') evenly distributed all-zero 
columns used as a jammer training period as well as the columns of $\bS=[\bS_T,\bS_D]$.\footnote{Spreading
the training period over the frame is the most sensible strategy against dynamic multi-antenna jammers 
as considered in \fref{sec:smart}, so this gives the strongest possible baseline.}
The submatrices of~$\bY$ consisting of the columns of the jammer training period, the pilot phase, and the data phase
are denoted $\bY_J\in\opC^{B\times R}$, $\bY_T\in\opC^{B\times T}$, and $\bY_D\in\opC^{B\times D}$, respectively. 
Our JMD methods will use $R=0$ since they do not require a jammer training period.
We assume a frame duration of $L=100$ channel uses. 
 
We consider QPSK transmission, and we use a square pilot matrix ($T=U=16$). 
The pilots are selected as rows of a $U\times U$ Haar matrix 
(normalized to unit symbol energy) and not revaled to the jammer.
The signal-to-noise ratio (SNR) is
\begin{align}
\textit{SNR} \define \frac{\Ex{\bms_k}{\|\bH\bms_k\|_s^2}}{\Ex{\bmn_k}{\|\bmn_k\|_2^2}} 
= \frac{\|\bH\|_F^2}{B\No}.
\end{align}
Furthermore, we characterize the receive power of the jammer interference relative to the 
receive power of the average UE as
\begin{align}
	\rho \define \frac{\frac{1}{L}\|\bJ\bW\|_F^2}{\frac1U\Ex{\bms_k}{\|\bH\bms_k\|_2^2}} 
	= \frac{\frac{1}{L}\|\bJ\bW\|_F^2}{\frac{1}{U}\|\bH\|_F^2},
\end{align}
where we deterministically scale the interference $\bJ\bW$ to some specified $\rho$. 
As performance metrics, we consider the uncoded bit error rate (BER) and 
a metric called modulation~error ratio (MER) between the data symbols $\bS_D$ 
and their estimate~$\hat\bS_D$, 
\begin{align}
	\textit{MER}\triangleq \sqrt{\frac{\mathbb{E}\big[\|\hat\bS_D - \bS_D\|_F^2\big]}{\mathbb{E}\big[\|\bS_D\|_F^2\big]}}.
\end{align}
We use the MER as a surrogate for error vector magnitude (EVM), which 
the 3GPP 5G NR technical specification requires to be below 17.5\% \cite[Tbl. 6.5.2.2-1]{3gpp21a} for QPSK transmission.

\subsection{Higher Data Rates}\label{sec:rate}

The first advantage of JMD compared to training-period based mitigation schemes 
is increased data rates since no symbol time slots are reserved for estimating the jammer's channel.

\vspace{2mm}
\subsubsection*{\textbf{Jammer Model}}
The rate advantage due to the absence of a training period is shown for the following jammer model:
\subsubsection*{\tinygraycircled{1} Single-antenna barrage jammer} 
One single-antenna jammer transmits i.i.d. circularly-symmetric complex
Gaussian noise during the entire frame. The jammer power is $\rho=30$\,dB.

\vspace{2mm}
\subsubsection*{\textbf{Baselines}} 
The following receivers are used as baselines:

\subsubsection*{POS-JED}
This receiver uses $0\leq R \leq L-T$ channel uses in each frame for estimating the jammer subspace
as the $I$ principal left singular vectors of $\bY_J$.  This reduces the number of channel uses for data transmission 
from $D=K-T$ to $D=K-T-R$.
The jammer is mitigated by first projecting~$\bY_T$ and $\bY_D$ onto the orthogonal complement of the estimated subspace, 
followed by performing FBS-based joint channel estimation and data detection analogous to MAED. 

\subsubsection*{G-POS-JED} 
This receiver serves as an upper bound to the achievable performance of MAED. 
It works analogous to MAED but is furnished with ground-truth knowledge of the jammer channel matrix $\bJ$. 
It therefore uses $R=0$ and sets~$\tilde\bP^{(t)}$ on line~7 of \fref{alg:maed} 
to the optimal projector $\bP = \bI_B - \Hj\pinv\Hj$.

\subsubsection*{POS-BOX}
Like POS-JED, this receiver uses $0\leq R \leq L-T$ channel uses of each frame for estimating the jammer subspace
as the $I$ principal left singular vectors of $\bY_J$. The jammer is mitigated by first projecting $\bY_T$ and $\bY_D$ 
onto the orthogonal complement of the estimated subspace, followed by performing LS channel estimation and FBS-based data 
detection with a box prior analogous to SANDMAN. 

\subsubsection*{G-POS-BOX}
This receiver serves as an upper bound to the achievable performance of SANDMAN. 
It works analogous to SANDMAN but is furnished with ground-truth knowledge of the jammer's channel matrix~$\bJ$. 
It therefore uses $R=0$ and fixes~$\tilde\bP^{(t)}$ on line 7 of \fref{alg:sandman} 
to the optimal projector $\bP = \bI_B - \Hj\pinv\Hj$.

All of the above algorithms run $t_\text{max}=30$ iterations.

\vspace{2mm}
\subsubsection*{\textbf{Results}}
In order to demonstrate the rate increase due to the absence of a jammer training period, we consider the tradeoff between 
the lowest SNR for which a receiver satisfies $\text{MER}\leq 17.5\%$ and the ratio
\begin{align}
	r = \frac{L-T-R}{L-T},
\end{align}
which is the number of channel uses $L-T-R$  per communication frame that are available for data transmission, 
normalized with respect to the maximum number of samples $L-T$ available without a jammer estimation phase.
The ratio~$r$ translates directly to the achieved data rate (measured in terms of bits/s).
\fref{fig:rate_reduction} shows the results:

Since SANDMAN and MAED as well as G-POS-BOX and G-POS-JED have no jammer training period ($R=0$), they achieve $r=1$. 
In contrast, the detection accuracy of \mbox{POS-BOX} and POS-JED increases with the redundancy $R$ 
(a longer training period yields a better estimate of the jammer channel, which enables more precise nulling), 
at the expense of~$r$. 
Our results demonstrate that SANDMAN and MAED, which can utilize the receive data of the entire frame to estimate the jammer subspace, 
achieve virtually the same performance as their genie-assisted counterparts G-POS-BOX and G-POS-JED. 
In contrast, POS-BOX and POS-JED can only use a subset of the receive samples to estimate the jammer interference, 
and so perform significantly worse even when dedicating a significant fraction of the frame to jammer training. 
In fact, POS-BOX does not reliably outperform SANDMAN even when using almost the entire frame for jammer training
(so that no almost no channel uses remain for data transmission, \mbox{$r\to0$}). The performance of POS-JED even 
deteriorates when $r\to0$ (since joint channel estimation and data detection degenerates when there are very 
few data symbols), and it is always at least $0.18$\,dB away from the performance of MAED.

\begin{figure}[tp]
\centering
\includegraphics[width=0.95\columnwidth]{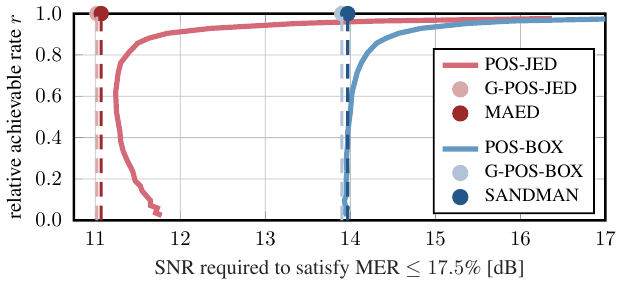}
\vspace{-2mm}
\caption{Trade-off between the relative achievable rate $r$ and the lowest SNR for which the different 
receivers satisfy the criterion $\text{MER}\leq17.5\%$ 
when mitigating a single-antenna barrage jammer \tinygraycircled{1}. 
}
\label{fig:rate_reduction}
\end{figure}

\begin{figure*}[tp]
\centering
\includegraphics[width=\textwidth]{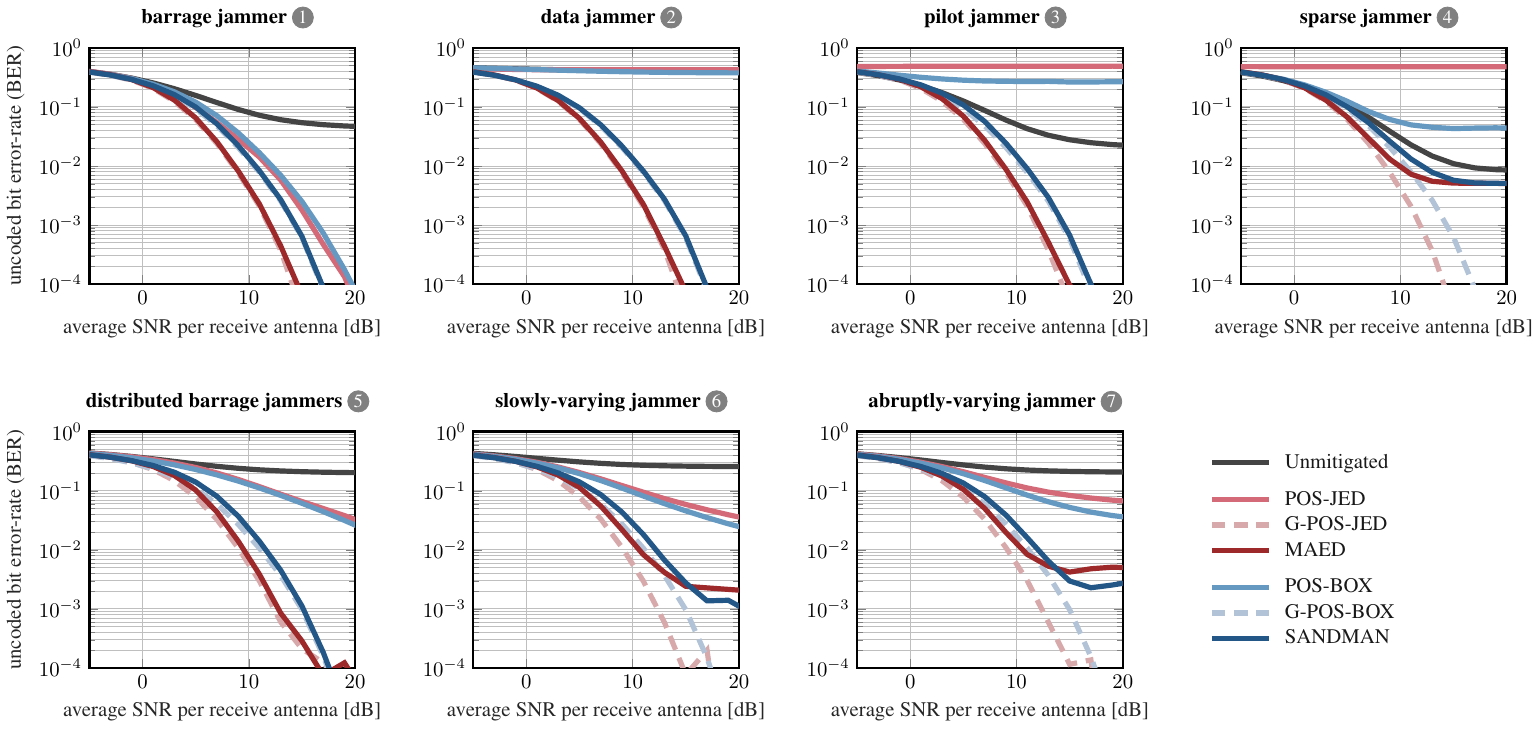}
\vspace{-5mm}
\caption{Uncoded bit error-rate (BER) vs. SNR performance of different receivers when mitigating different kinds of jammers, 
including smart, distributed, and dynamic multi-antenna jammers.}
\vspace{-4mm}
\label{fig:smart}
\end{figure*}

\subsection{Mitigating Smart, Distributed, and Multi-Antenna Jammers} \label{sec:smart}
In this experiment, we analyze the ability of JMD to mitigate smart jammers that might not jam during 
any training phase but only at specific instants, or spatially distributed single-antenna jammers, 
or multi-antenna jammers that use dynamic beamforming to change their interference subspace. 
\vspace{2mm}
\subsubsection*{\textbf{Jammer Models}}
Besides the single-antenna barrage jammer~(\tinygraycircled{1}), we consider the following types of jammers:
\subsubsection*{\tinygraycircled{2} Single-antenna data jammer} A single-antenna jammer
that does not jam the training period and pilot phase, 
and that transmits i.i.d. complex Gaussian noise with \mbox{$\rho=30$\,dB} in the data phase.
Note that this jammer is smart (i.e., \mbox{protocol-aware}). 
\subsubsection*{\tinygraycircled{3} Single-antenna pilot jammer} A single-antenna jammer
that does not jam the training period and data phase, 
and that transmits i.i.d. complex Gaussian noise with \mbox{$\rho=30$\,dB} during the pilot phase.
Note that this jammer is smart, too.
\subsubsection*{\tinygraycircled{4} Single-antenna sparse jammer} A non-smart single-antenna jammer
that only jams for a (randomly selected) single symbol time per communication frame, with \mbox{$\rho=30$\,dB.}
\subsubsection*{\tinygraycircled{5} Distributed single-antenna barrage jammers} 
Four distributed and statistically independent single-antenna jammers transmit
i.i.d. complex Gaussian noise with $\rho=30$\,dB. 
\subsubsection*{\tinygraycircled{6} Slowly varying multi-antenna jammer}
A non-smart four-antenna jammer where the $\tilde\bmw_k$ (see \eqref{eq:jammer_beamforming}) are white Gaussian random vectors 
with $\rho=30$\,dB. The matrices $\bA_k$ are constructed as follows:
Only the leftmost column $\bma_{k,1}$ of $\bA_k$ is nonzero. 
For randomly selected instants $k_1<\dots<k_M, M=5$, 
the vector $\bma_{k_m,1}$ is fixed to randomly drawn vectors $\{\bma^{(m)}\}_{m=1}^M$. 
For $k_m\!<\!k\!<\!k_{m+1}$, $\bma_{k,1}$ interpolates between $\bma^{(m)}$ and $\bma^{(m+1)}$.
\subsubsection*{\tinygraycircled{7} Abruptly varying multi-antenna jammer}
A non-smart four-antenna jammer that transmits
only~with a random subset of one or two of its antennas at each instant $k$. 
To this end, a randomly selected subset of the rows
of the beamforming matrix $\bA_k$ (see \eqref{eq:jammer_beamforming}) is populated with i.i.d. $\setC\setN(0,1)$ entries, 
and the other rows are set to zero. The beamforming matrix $\bA_{k+1}$ is equal to $\bA_k$ with probability $0.95$ and
otherwise is redrawn (including the support set of the rows of $\bA_{k+1}$). 
The vectors $\tilde\bmw_k$ are white Gaussian random vectors with $\rho=30\,$dB.

\vspace{2mm}
\subsubsection*{\textbf{Baselines}} 
We compare SANDMAN and MAED with the baselines from \fref{sec:rate}. 
The training period-based methods POS-BOX and POS-JED use a redundancy of $R=4$ and $D=80$ data transmission symbols
per frame (yielding a relative rate of $r=80/84\approx0.95$) 
while the other methods use $R=0$ and $D=84$ data transmission symbols (yielding $r=1$). 
All iterative algorithms run for $t_\text{max}=30$ iterations when facing a single-antenna jammer, 
and for $t_\text{max}=50$ iterations when facing distributed or multi-antenna jammers. 
The algorithms are provided with knowledge of the number $I$ of jammer antennas.
Other than that, the configurations of the algorithms do not depend on the type of jammer being faced. 
For context, we also include the following unmitigated receiver as an additional baseline:
\subsubsection*{Unmitigated} This receiver does not mitigate the jammer. 
It performs LS channel estimation and (jammer-oblivious) LMMSE data detection.

\vspace{2mm}
\subsubsection*{\textbf{Results}}
\fref{fig:smart} depicts the results. 
The performance of the G-POS-BOX and G-POS-JED baselines depends on the jammer type only via 
the number of jammer antennas.\footnote{Since the jammer is perfectly nulled using ground-truth knowledge, 
there is no residual jamming interference. However, each jammer antenna entails the loss of one degree of 
freedom after nulling, see \cite[Prop. 3]{marti2024fundamental}.}
In contrast, the performance of the unmitigated receiver suffers to varying degrees---but always significantly---under 
the different types of jammers: the least harmful jammers for an unmitigated receiver are the pilot jammer and the sparse 
jammer (with a BER of around $1\%$ at high SNR), and the most harmful are the data jammer \tinygraycircled{2}
as well as distributed and multi-antenna jammers \tinygraycircled{5}, \tinygraycircled{6}, \tinygraycircled{7} 
(with a BER of $20\%$\,-\,$50\%$ even at high~SNR). 
As expected, the training period-based baselines POS-BOX and POS-JED  mitigate the barrage jammers \tinygraycircled{1}
and \tinygraycircled{5} successfully, 
but fail against all smart or dynamic jammers.\footnote{POS-BOX and POS-JED often perform even worse 
than the unmitigated method (see, e.g., the pilot jammer \tinygraycircled{3}). 
This is due to the fact that the~nonlinear data detectors of POS-BOX and POS-JED try to fit the receive signal
with a signal model that is accurate only if the jammer subspace is correctly estimated, and which otherwise can 
fail catastrophically.
In contrast, the linear detection of the unmitigated baseline is less susceptible to model mismatches.
} 

In contrast, the JMD-type methods SANDMAN and MAED are able to mitigate \emph{all} jammers more or less successfully: 
\begin{itemize}
	\item They achieve virtually the same performance as their genie-assisted counterparts G-POS-BOX and G-POS-JED
	against the barrage jammers \tinygraycircled{1} and \tinygraycircled{5} as well as against the data jammer \tinygraycircled{2}
	and the pilot jammer \tinygraycircled{3}. This indicates that SANDMAN and MAED null these jammers essentially perfectly.
	Note also that SANDMAN and MAED outperform their training period-based counterparts POS-BOX and POS-JED even against
	the barrage jammers \tinygraycircled{1} and \tinygraycircled{5} (as expected based on \fref{sec:rate}).  
	\item Also against the smart jammers \tinygraycircled{2} and \tinygraycircled{3}, 
	SANDMAN and MAED achieve virtually the same performance as their genie-assisted counterparts. In contrast, 
	their training period-based counterparts fail completely against such smart jammers that do not jam during the 
	training~period. 
	\item Even against the sparse jammer \tinygraycircled{4}, which 	might be expected to be their Achilles heel (based on the
	discussion in \fref{sec:jmd} and \fref{sec:chest}), SANDMAN and MAED perform outperform the unmitigated baseline 
	as well as their training-period based counterparts. They manage to achieve a BER of below $1\%$ at high SNR. 
	\item Against the dynamic multi-antenna jammers \tinygraycircled{6} and~\tinygraycircled{7}, SANDMAN and MAED
	achieve BERs significantly below $1\%$ at high SNR. In contrast, the BERs of their training period-based counterparts 
	POS-BOX and POS-JED remain significantly above $1\%$ even against such non-smart dynamic jammers. 
	Note, however, that also SANDMAN and MAED eventualy hit an error floor due to difficulty of disentangling a
	high-dimensional and time-varying interference subspace from the signal subspace. The error floor
	of SANDMAN is lower than that of MAED because the optimization problem approximated by SANDMAN is less
	complex, and thus less prone to misconvergence, than the one approximated by MAED.
\end{itemize}
In summary, SANDMAN and MAED outperform their training period-based counterparts POS-BOX and POS-JED for every 
considered jammer type (sometimes decivisely). They also consistently outperform the unmitigated baseline. 
This demonstrates the suitability of JMD-based methods for mitigating smart and/or dynamic single-antenna jammers, 
distributed jammers, or multi-antenna jammers.

\section{Discussion and Conclusion}

We have proposed joint jammer mitigation and data detection (JMD), a novel paradigm 
for mitigating jammers in MIMO systems that does not use a dedicated jammer training period. 
As a result, JMD is able to mitigate smart jammers regardless of (i) when they are active
and (ii) how they vary their multi-antenna transmit beamforming. 
We have provided theoretical success guarantees for the case of smart single-antenna jammers, 
and we have proposed two JMD-type algorithms (SANDMAN and MAED) whose efficacy against single- 
and multi-antenna jammers has been demonstrated through extensive simulations. 

At the moment, JMD still exhibits certain drawbacks:
(i)~we only provide success guarantees for single-antenna jammers, 
(ii)~JMD requires the number of jammer antennas (or the rank of jamming interference) to be known at the receiver in advance, 
and (iii)~the JMD-type algorithms SANDMAN and MAED exhibit an error floor against sparse jammers 
(which are prone to eclipsing, see \fref{sec:jmd} and \fref{sec:chest}) and against multi-antenna jammers 
(for which solving the respective optimization problems is challenging). 
All of these issues can be remedied by combining JMD with a recently developed technique in which 
the transmit signals are transformed using a linear time-domain transform. 
Preliminary results are shown in \cite{marti2023universal} and will be detailed more fully in future work.

\appendix 
\label{app:test}

\section*{Proofs}
\label{app:proofs}

\subsection{Proof of \fref{thm:perfect_csi:correct}}
We start by noting that the Frobenius norm in \eqref{eq:jmd_problem} is nonnegative and evaluates to zero if 
and only if its argument is the all-zero matrix. 
That $\{\hat\bP,\hat\bS_D\} = \{\bI_B - \bJ\pinv{\bJ}, \bS_D\}$ is \emph{a} 
minimizer follows then directly from 
\begin{align}
	\!\!\hat\bP(\bY\!_D - \bH\hat\bS_D)
	&= (\bI_B - \bJ\pinv{\bJ}) (\bH\bS_D + \bJ\bW\!_D - \bH\bS_D)\!\! \\
	&= (\bI_B - \bJ\pinv{\bJ})\bJ\bW\!_D \\
	&= \mathbf{0},
\end{align}
since $(\bI_B - \bJ\pinv{\bJ})\bJ = \mathbf{0}$. It remains to show that $\{\hat\bP,\hat\bS_D\}$
is the \emph{only} minimizer of \eqref{eq:jmd_problem}. For this, we rewrite the argument 
of the Frobenius norm in \eqref{eq:jmd_problem} as
\begin{align}
	\tilde\bP(\bH\bS_D + \bJ\bW\!_D - \bH\tilde\bS_D) 
	&= \tilde\bP \begin{bmatrix}\bH,\bJ \end{bmatrix} \begin{bmatrix} \bS_D - \tilde\bS_D\\ \bW\!_D \end{bmatrix}.
	\label{eq:data_objective_rewritten}
\end{align}
Since $\tilde\bP\in\mathscr{G}_{B-I}(\opC^B)$, it can null a matrix only if that matrix
has rank less than or equal to $I$. But by assumption, the matrix $[\bH,\bJ]$ has full column rank $U+I$. 
So its product with $[\bS_D - \tilde\bS_D; \bW\!_D]$ can give a matrix whose rank does not exceed $I$
only if the rank of $[\bS_D - \tilde\bS_D; \bW\!_D]$ itself does not exceed $I$. 
By our assumption that the jammer is not eclipsed, this can be the case only if $\tilde\bS_D = \bS_D$, 
in which case \eqref{eq:data_objective_rewritten} simplifies to $\tilde\bP\bJ\bW\!_D$. 
Since the jammer is not eclipsed, the rank of $\bW\!_D$ itself is equal to $I$. 
Hence $\textit{col}(\bJ\bW\!_D)=\textit{col}(\bJ)$, so that $\tilde\bP\bJ\bW\!_D=\mathbf{0}$
if and only if $\tilde\bP=\bI_B - \bJ\pinv{\bJ}$.
\hfill$\blacksquare$

\subsection{Proof of \fref{thm:csi_eclipsing}}

Since signals from the pilot phase play no role for this result, it will be convenient 
to simply call $\bmw_D$ the jammer transmit signal (instead of ``the transmit signal during the data phase'').

We start by defining the \emph{coset} 
\begin{align}
	\mathfrak{E}(\bS_D) \triangleq \Big\{ \bE(\tilde\bS_D; \bS_D)=\bS_D - \tilde\bS_D : 
	\tilde\bS_D \in \setS^{U\times D}\!\setminus\!\{\bS_D\} \Big\}.
\end{align}
Note that, by definition, the coset $\mathfrak{E}(\bS_D)$ does not include the all-zero matrix.
We can now rewrite \fref{def:csi_eclipsing} as follows: 
\begin{defi}[Eclipsing with perfect CSI]
A single-antenna jammer is eclipsed in a given frame if there exists a matrix 
$\bE(\tilde\bS_D; \bS_D)\in\mathfrak{E}(\bS_D)$ such that 
$[\bE(\tilde\bS_D; \bS_D); \tp{\bmw_D}]$ is a matrix of rank $1$. 
\end{defi}

Thus, the jammer eclipses if (i) $\mathfrak{E}(\bS_D)$ includes a matrix $\bE(\tilde\bS_D; \bS_D)$ whose 
rows are all collinear 
and (ii) $\tp{\bmw_D}$ is collinear with these~rows. 
Without loss of generality, for all $\bE(\tilde\bS_D; \bS_D)\in\mathfrak{E}(\bS_D)$, 
denote the $u$th row of $\bE(\tilde\bS_D; \bS_D)$ by $\tp{\bme_{(u)}(\tilde\bms_{(u)};\bms_{(u)})}$,
where $\tp{\tilde\bms_{(u)}}$ and $\tp{\bms_{(u)}}$ denote the corresponding rows of $\tilde\bS_D$ and $\bS_D$, respectively. 
We then define the cosets
\begin{align}
	\mathfrak{e}_{(u)}(\bms_{(u)}) \!\triangleq \!
	\big\{\bme_{(u)}(\tilde\bms_{(u)};\bms_{(u)})\!=\tilde\bms_{(u)}\!-\bms_{(u)} : \tilde\bms_{(u)}\!\!\in\!\setS^D\!\setminus\!\{\bms_{(u)}\} \!\big\}
\end{align}
for $u=1,\dots,U$. 
Note that, by definition, $\mathfrak{e}_{(u)}(\bms_{(u)})$ does not include the all-zero vector.
We have the following lemmas:

\begin{lem} \label{lem:coset_eclipsing}
	A single-antenna jammer with transmit signal $\bmw_D$ eclipses if and only if, for some $u\in[1:U]$, there exists 
	an $\bme\in\mathfrak{e}_{(u)}(\bms_{(u)})$ that is collinear with $\bmw_D$.
\end{lem}

\begin{proof}
We start with the ``only if'' direction: 
The $\mathfrak{e}_{(u)}(\bms_{(u)})$, $u=[1:U],$ contain the rows of the elements of $\mathfrak{E}(\bS_D)$ 
which are distinct from zero. Thus, if there exists no $\bme\in\mathfrak{e}_{(u)}(\bms_{(u)})$ that is collinear
with $\bmw_D$---for any $u\in[1:U]$---then for any matrix~$\bE$ which contains at least one row from 
$\bigcup_{u=1}^U \mathfrak{e}_{(u)}(\bms_{(u)})$, the augmented matrix $[\bE; \tp{\bmw_D}]$ has at least least rank $2$. 
Since all matrices $\mathbf{E}\in\mathfrak{E}(\bS_D)$ contain at least one such row, 
it follows that the jammer does not eclipse. 

For the ``if'' direction, we may assume that there exists~a $u\in[1:U]$ and an $\bme\in\mathfrak{e}_{(u)}(\bms)$ 
such that $\bmw_D$ is collinear with~$\bme$. We now construct a matrix $\tilde\bS_D\in\setS^{U\times D}$ as follows: 
In all rows except the $u$th one, $\tilde\bS_D$ equals $\bS_D$; 
in the $u$th one, $\tilde\bS_D$ equals the corresponding row of $\bS_D$ \emph{minus}~$\tp{\bme}$. 
By being constructed this way, $\tilde\bS_D$ has the properties that 
(i) $\bS_D-\tilde\bS_D\in\mathfrak{E}(\bS_D)$, and 
(ii) $[\bS_D-\tilde\bS_D;\tp{\bmw_D}]=[\tp{\mathbf{0}};\dots;\tp{\mathbf{0}};\tp{\bme};\tp{\mathbf{0}};\dots;\tp{\mathbf{0}};\tp{\bmw_D}]$ has rank one. 
Thus, the jammer is eclipsed. 
\end{proof}

\begin{lem} \label{lem:multiple_eclipsing_chances}

	Let $\bmw_D$ be the transmit signal of a single-antenna jammer,
	let $\bms\sim\textnormal{Unif}[\setS^D]$ be a random vector, 
	and define
	\begin{align}
	\mathfrak{e}(\bms) \triangleq \big\{\tilde\bms-\bms : \tilde\bms\in\setS^D\setminus\{\bms\} \big\}. \label{eq:ecoset}
	\end{align}
	Let furthermore $u\in[1:U]$ be fixed. 
	Then the following two events have identical probabilities:
	\begin{enumerate}
		\item[(i)\,] There is a $\bme\!\in\!\mathfrak{e}_{(u)}(\bms_{(u)})$ such that $\bmw_D$ is collinear with~$\bme$.
		\item[(ii)] There is a $\bme\!\in\!\mathfrak{e}(\bms)$ such that $\bmw_D$ is collinear with $\bme$.
	\end{enumerate}	
\end{lem}

\begin{proof}
	Since $\bms_{(u)}\sim\textnormal{Unif}[\setS^D]$ for any $u\in[1:U]$ (which follows since 
	the transmit symbols are assumed to be i.i.d. uniform, cf. \fref{sec:model})	, the random sets 
	$\mathfrak{e}_{(u)}(\bms_{(u)})$ and $\mathfrak{e}(\bms)$ have the same distribution. The result follows immediately.
\end{proof}

\fref{lem:coset_eclipsing} and \fref{lem:multiple_eclipsing_chances} give rise to the following corollary:

\begin{cor} \label{cor:eclipsing}
	Let $\bmw_D$ be the transmit signal of a single-antenna jammer. 
	If $\bms\sim\textnormal{Unif}[\setS^D]$ and $\mathfrak{e}(\bms)$ is defined as in~\eqref{eq:ecoset}, 
	and if 
	\begin{align}
		\!\!\! q(\bmw_D) \!\triangleq\! \Prob(\exists \bme\in\mathfrak{e}(\bms) \text{ such that } \bmw_D \text{ is collinear with } \bme),\!\!
		\label{eq:define_q}
	\end{align}
	then the jammer eclipses with probability $1-(1-q(\bmw_D))^U$. 
\end{cor}
\begin{proof}
	By \fref{lem:coset_eclipsing}, the jammer eclipses if and only if there exists 
	a $\bme\in\mathfrak{e}_{(u)}(\bms_{(u)})$ that is collinear with $\bmw_D$, for some $u\in[1:U]$.
	Eclipsing corresponds therefore to the union of these $U$ events.
	Since the $U$ rows of $\bS_D$ are independent (cf.~\fref{sec:model}), these events are independent, 
	and by \fref{lem:multiple_eclipsing_chances}, they all have probability $q(\bmw_D)$. 
	The result follows immediately.
\end{proof}

To complete the proof of \fref{thm:csi_eclipsing}, it now only remains to show that the probability 
$q(\bmw_D)$ as defined in \eqref{eq:define_q} is at most
\begin{align}
	\frac{2^{\|\bmw_D\|_0}-1}{4^{\|\bmw_D\|_0-1}}.
\end{align}
For this, we define the concept of an \emph{eclipsing-optimal jammer}.

\begin{defi}
	An \emph{eclipsing-optimal jammer} is a jammer which, for a given zero-norm $\|\bmw_D\|_0$, transmits
	a signal with maximal probability of eclipsing, i.e., a signal from the set
	\begin{align}
		\argmax_{\bmw'\in\opC^D:\|\bmw'\|_0=\|\bmw_D\|_0} q(\bmw').
	\end{align}	
\end{defi}

\begin{lem} \label{lem:jammer_constellation1}
To be eclipsing-optimal, a single-antenna jammer has to transmit a signal 
$\bmw_D\in\hat\setW_\alpha^D$ for~some $\alpha\in\opC$, 
where $\setW_\alpha\triangleq \{0, \alpha , \alpha i, -\alpha , -\alpha i\}$. 
\end{lem}

\begin{proof}
By \fref{cor:eclipsing}, for $\bms\sim\text{Unif}[\setS^D]$, the probability of eclipsing is increasing in 
the probability that there exists an $\bme\in\mathfrak{e}(\bms)$
such that $\bmw$ is collinear with $\bme$. 
The entries of any element $\bme\in\mathfrak{e}(\bms)$ take value in the 
set $\Delta\setS \triangleq \{s-\tilde s : s,\tilde s \in \setS\}$ illustrated in \fref{fig:delta_quemlis}. 
We have 
\begin{align}
	\!\!\!\!\Delta\setS =& \hat\setW_{\sqrt2} \nonumber \\
	&\! \cup \big\{\sqrt2(1\!+\!i), \sqrt2(1\!-\!i), \sqrt2(-1\!+\!i), \sqrt2(-1\!-\!i)\big\}.\!\!\!
\end{align}
So the $\bmw_D\in\opC^D$ that lead to nonzero probability of eclipsing are exactly those 
contained in $\alpha\Delta\setS^D\triangleq\{\alpha\Delta\bms:\Delta\bms\in\Delta\setS^D\}$ for some $\alpha\in\opC$. 
However, not all $\bme\in\Delta\setS^D$ are contained in $\mathfrak{e}(\bms)$ with equal probability, 
and so not all $\bmw_D\in\alpha\Delta\setS^D$ lead to the same probability of eclipsing. 
Define
\begin{align}
	\mathfrak{e}_k(s_k) \triangleq  \{e_k : \bme\in\mathfrak{e}(\bms) \} \label{eq:coset_entries}
\end{align}
to be the $k$th ``entry'' of $\mathfrak{e}(\bms)$, which depends on $\bms$ only through~$s_k$.
Vice versa, $\mathfrak{e}(\bms)=\{\bme:e_k\in\mathfrak{e}_k(s_k),k=1,\dots,D\}$. 
For all $\bme\in\Delta\setS^D\setminus\{\boldsymbol{0}\}$, we have
\begin{align}
	\Pr(\bme\in\mathfrak{e}(\bms)) = \prod_{k=1}^D\Pr(e_k\in\mathfrak{e}_k(s_k)). \label{eq:e_prod}
\end{align}
The reason why not all $\bmw_D\in\alpha\Delta\setS^D$ have the same probability of eclipsing
is that, for given~$k$, $\mathfrak{e}_k(s_k)$ does not contain all elements of $\Delta\setS$
with equal probability,~cf.~\fref{fig:delta_quemlis}:

First, we have $P(0\in\mathfrak{e}_k(s_k))=1$. That is, the origin is contained in
 $\mathfrak{e}_k(s_k)$ with probability one (i.e, for all possible realizations of $s_k$), 
as it can be ``reached'' (by subtracting some $\tilde s_k\in\setS$) 
from any symbol $s_k\in\setS$, by setting $\tilde s_k = s_k$.\footnote{
This is also the reason why \eqref{eq:e_prod} does not hold for $\bme=\boldsymbol{0}$: 
the right-hand-side is equal to $1$, but the left-hand-side is equal to $0$, cf. \eqref{eq:ecoset}.}
Note that by transmitting a zero, the jammer does not increase $\|\bmw\|_0$. 

The four points $\sqrt2, i\sqrt2, -\sqrt2, -i\sqrt2$ on the coordinate axes are each contained in 
$\mathfrak{e}_k(s_k)$ with probability $\frac12$, as they can be reached from $s_k$ if $s_k$ takes on 
either of the two constellation values adjacent to that value of the coordinate semi-axis.

Finally, the four corner points $\sqrt2, i\sqrt2, -\sqrt2, -i\sqrt2$ are each contained 
in $\mathfrak{e}_k(s_k)$ with probability $\frac14$, as they can be reached from $s_k$ only if 
$s_k$ takes on the constellation value in the corresponding quadrant. 

So, to be optimal for some $\|\bmw\|_0\geq0$, a jammer should~restrict its transmit alphabet
to the zero symbol and (a scaled version of) the four points $\sqrt2, i\sqrt2, -\sqrt2, -i\sqrt2$.
That is, the jammer should restrict its alphabet to $\hat\setW_\alpha$ for some~$\alpha\in\opC$.
\end{proof}

\begin{figure}[tp]
\centering
\hspace{8mm}
\includegraphics[height=4.5cm]{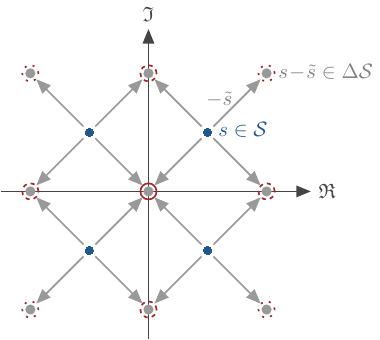}
\caption{Illustration of $\setS$ and $\Delta\setS$. Values in $\Delta\setS$ which are contained
in $\mathfrak{e}_k(s_k)$ with probability $1$ are circumscribed with a drawn red line; 
values that are contained with probability $\frac12$ are circumscribed with a dashed red line; 
values that are contained with probability $\frac14$ are circumscribed with a dotted red line.}
\label{fig:delta_quemlis}
\end{figure}

\begin{lem} \label{lem:prob_of_q}
For an optimal jammer, the probability $q$ as defined in \eqref{eq:define_q} is equal to
\begin{align}
	q(\bmw_D) = 
	\begin{cases}
		1 &\text{if } \bmw_D=\mathbf{0} \\
		\frac{2^{\|\bmw_D\|_0}-1}{4^{\|\bmw_D\|_0-1}} &\text{else}.
	\end{cases}
\end{align}
\end{lem}

\begin{proof}
Without loss of generality, an optimal jammer uses the transmit alphabet $\hat\setW_{\sqrt{2}}=\{0,\sqrt{2},i\sqrt{2},-\sqrt{2},-i\sqrt{2}\}$.
First, if the jammer only transmits zeros, $\bmw_D=\mathbf{0}$, then $\bmw_D$ is always collinear with any 
$\bme\in\mathfrak{e}(\bms)$, and hence $q(\bmw_D)=1$. In contrast, if $\bmw_D\neq\mathbf{0}$, then we 
can ignore the zeros in the jammer's transmit signal for our analysis:
Let $\bmw\in\opC^L$ be a jammer transmit signal (of arbitrary length $L$) without zeros, and 
if $\bmw'\in\opC^{L'}$ with $L'>L$ is a zero-padded version of $\bmw$, with
$L'-L$ zeros inserted at arbitrary indices $\setI\subset[1:L']$. 
If and only if~$\bmw$ is collinear with some $\bme\in\mathfrak{e}(\bms)$, then $\bmw'$ 
is collinear with $\bme'$, where $\bme'\in\opC^{L'}$ is the zero-padded version 
of $\bme$ (with zeros inserted at indices $\setI$) and is contained in $\mathfrak{e}(\bms')$ 
for any $\bms'\in\setS^{L'}$ that satisfies $\bms'_{[\setI]}=\bms$. 
So we have both $q(\bmw)=q(\bmw')$ and $\|\bmw\|_0=\|\bmw'\|_0$. 
It follows that we can ignore the presence of zero-valued entries in the jammer's transmit signal.

Furthermore, note that the $\mathfrak{e}_k(s_k)$ as defined in \eqref{eq:coset_entries} are independent 
and identically distributed for all $k$ (since the $s_k$ are i.i.d.). 
Therefore, and since the $s_k$ are uniformly distributed, we can assume
without loss of generality that all non-zero symbols that the jammer transmits are 
equal to $\sqrt2$. 

If the jammer only transmits one non-zero symbol, $w_1=\sqrt{2}$,
then for any realization of $s_1$, we have eclipsing: 
An $\bme\in\mathfrak{e}(\bms)$ that is collinear with $\bmw$ can always be found by setting
$\tilde s_1\in\setS\setminus\{s_1\}$ and $\tilde s_k = s_k$ for any $k\neq1$.

If the jammer transmits two or more nonzero symbols $w_\ell=\sqrt2$, $\ell=1,\dots,L$, $L\geq2$,  
then we only have eclipsing if all $s_\ell, \ell=1,\dots,L$ lie on the same half-plane of the constellation. 
More precisely: An $\bme\in\mathfrak{e}(\bms)$ that is collinear with $[w_1,\dots,w_{L}]$ exists if and only if
all $s_\ell, \ell=1,\dots,L$ lie on the same half-plane of the constellation. 
To get the probability of this event, we can simply count the allowable combinations, 
since all realizations of $\bms$ are equally likely: 
There are four different half-planes (left, right, top, bottom), and in each half-plane, 
there are $2^{L}=2^{\|\bmw_D\|_0}$ different sequences.
But in counting like this, we count each of the four constant sequences
$e^{i\frac{\pi}{4}}\mathbf{1}, e^{i\frac{3\pi}{4}}\mathbf{1}, e^{i\frac{5\pi}{4}}\mathbf{1},
e^{i\frac{7\pi}{4}}\mathbf{1}$ 
twice (e.g., $e^{i\frac{\pi}{4}}\mathbf{1}$ is counted both for the right and the top half-plane), 
so we need to subtract these four double-counted realizations. 
This gives us $4\cdot2^{\|\bmw_D\|_0}-4$ realizations of $\bms$ that lead to eclipsing, 
out of $4^{\|\bmw_D\|_0}$ sequences in total. So the probability $q(\bmw_D)$ is
\begin{align}
	q(\bmw_D) = \frac{4\cdot2^{\|\bmw\|_0}-4}{4^{\|\bmw\|_0}} = \frac{2^{\|\bmw_D\|_0}-1}{4^{\|\bmw_D\|_0-1}}.
\end{align}
\end{proof}
\fref{thm:csi_eclipsing} now follows from \fref{cor:eclipsing} and \fref{lem:prob_of_q},
where the approximation for large $\|\bmw\|_0$ follows by approximating $q(\bmw_D)$ as $4\cdot2^{-\|\bmw\|_0}$
and then using a first-order Taylor approximation around $q(\bmw_D)=0$.
\hfill$\blacksquare$

\subsection{Proof of \fref{thm:independent_chest:correct}}
Using the theorem's assumptions and \eqref{eq:pilot_contamination}, we rewrite the argument of the optimization objective 
in \eqref{eq:jmd_problem2} as 
\begin{align}
	& \tilde\bP (\bY\!_D - \hat\bH \tilde\bS_D) \\
	&= \tilde\bP (\bH\bS_D + \bJ\bW\!_D - (\bH + \bJ\bW_T\pinv{\bS_T}) \tilde\bS_D) \\
	&= \tilde\bP \begin{bmatrix}\bH,\bJ \end{bmatrix} 
	\begin{bmatrix} \bS_D - \tilde\bS_D\\ \bW\!_D - \bW_T\pinv{\bS_T}\tilde\bS_D \end{bmatrix}
\end{align}
From here on, the proof is very similar to the one of \fref{thm:perfect_csi:correct} (considering the modified 
notion of eclipsing as defined in \fref{def:eclipse_chest}), and so is omitted. 
\hfill$\blacksquare$

\subsection{Proof of \fref{thm:independent_chest_eclipsing_nopilot}}

Assume without loss of generality that the jammer knows the pilot sequence of the first UE, 
i.e., the first row of $\bS_T$, which we denote $\tp{\bms_{T,1}}$. Then the jammer can transmit
$\bmw_T=\bms_{T,1}$ and $\bmw_D\in\setS^D$. 
We therefore have $\tp{\bmw_T}\pinv{\bS_T}=[1,0,\dots,0]$ and so 
\begin{align}
	\bSigma &= [\bS_D-\tilde\bS_D;\tp{\bmw_D} - \tp{\bmw_T}\pinv{\bS_T}\tilde\bS_D] \\
	&= [\bS_D-\tilde\bS_D;\tp{\bmw_D} - \tp{\tilde\bms_{D,1}}],
\end{align}
where $\tp{\tilde\bms_{D,1}}$ is the first row of $\tilde\bS_D$. 
In that case, consider for $\tilde\bS_D$ the matrix which is equal to $\bS_D$ on all 
rows except the first row, where it is equal to $\tp{\bmw_D}$. 
If $\tp{\bmw_D}\neq\tp{\bms_{D,1}}$, then $\tilde\bS_D\neq\bS_D$ and $\bSigma$ is a 
matrix of rank one, meaning that the jammer eclipses. 

Since $\tp{\bms_{D,1}}$ is drawn uniform at random from $\setS$, the probability that $\tp{\bmw_D}\neq\tp{\bms_{D,1}}$
is equal to $1-4^{-D}$. 
\hfill$\blacksquare$

\subsection{Proof of \fref{thm:independent_chest_eclipsing}}

If the jammer does not jam during the pilot phase, $\bmw_T=\mathbf{0}$, then eclipsing with channel estimation
coincides with eclipsing with perfect CSI (cf.\,\fref{def:csi_eclipsing}), and the result follows from~\fref{thm:csi_eclipsing}.

If the jammer does jam during the pilot phase, $\bmw_T\neq\mathbf{0}$, then $\bmw_T$ is independent of $\bS_T$
(since the jammer does not know $\bS_T$). 
We now focus on the last row of $\bSigma$, which is $\tp{\bmw_D}-\tp{\bmw_T}\pinv{\bS_T}\tilde\bS_D$.
Since $\bS_T$ is Haar distributed (and therefore unitary) up to a scale-factor, 
$\pinv{\bS_T}=\frac{1}{T}\herm{\bS_T}$ is Haar distributed up to a scale-factor, too. 
Hence, $\tp{\bmx}\triangleq\tp{\bmw_T}\pinv{\bS_T}$ is uniformly distributed over the complex $U$-dimensional 
sphere of radius $\|\bmw_T\|_2/T$ \cite[p.\,16]{meckes2019random}. 
We can therefore also write $\bmx = \frac{\|\bmw_T\|_2}{T \|\bmz\|_2}\bmz$, where $\bmz\sim\setC\setN(\mathbf{0},\bI_U)$, 
and rewrite the last row of $\bSigma$ as 
$\tp{\bmw_D} - \frac{\|\bmw_T\|_2}{T \|\bmz\|_2}\tp{\bmz}\tilde\bS_D$.
Here, $\tp{\bmz}\tilde\bS_D$ is a (transposed) complex Gaussian vector with covariance matrix 
$\herm{\tilde\bS_D}\tilde\bS_D$, so that its entries are complex circularly-symmetric scalar Gaussians with 
variance $U$ (corresponding to the energy of the columns of~$\tilde\bS_D$).
Hence, the last row $\tp{\bmw_D}-\tp{\bmw_T}\pinv{\bS_T}\tilde\bS_D$ has full support $D$ with probability one. 
Since this row does not depend on $\bS_D$, we can simply refer to \fref{thm:csi_eclipsing}, with 
the full-support (with probability one) vector $\tp{\bmw_D}-\tp{\bmw_T}\pinv{\bS_T}\tilde\bS_D$ \emph{in lieu} 
of the potentially sparse vector $\tp{\bmw_D}$, and the result follows. 
\hfill$\blacksquare$

\subsection{Proof of \fref{thm:joint_chest:correct}}
This theorem was already proved for the special case of a single-antenna jammer in \cite[Thm.\,1]{marti2023maed}. 
The extension to multi-antenna jammers is straightforward using arguments as in the proofs of 
\fref{thm:perfect_csi:correct} and \fref{thm:independent_chest:correct} and so is omitted. 
\hfill$\blacksquare$ 

\subsection{Proof of \fref{thm:convex}}
We first prove convexity in $\tilde\bS_D$ for fixed $\tilde\bP$:
For fixed $\tilde\bP$, the objective in \eqref{eq:sandman_problem} can be rewritten as 
\begin{align}
	\!\!\! \hat f(\tilde\bS_D) \triangleq& \Tr\big(\herm{\bY\!_D}\tilde\bP\bY\!_D \nonumber \\
	& + \herm{\tilde\bS_D}[\herm{\hat\bH}\tilde\bP\hat\bH -\alpha \bI_U]\tilde\bS_D 
	- 2\Re\{\herm{\bY\!_D}\tilde\bP\hat\bH\tilde\bS_D\}\big).\!\! \label{eq:traceform}
\end{align}
$\hat f$ is a real-valued function of the complex matrix $\tilde\bS_D$.
Following \cite{hjorungnes11a}, we represent the dependence on this complex input matrix by using 
both $\tilde\bS_D$ and $\tilde\bS_D^\ast$, giving rise to the four different Hessian matrices
$\setH_{\tilde\bS_D, \tilde\bS_D^\ast} \hat f, \setH_{\tilde\bS_D^\ast, \tilde\bS_D^\ast} \hat f, 
\setH_{\tilde\bS_D, \tilde\bS_D} \hat f$, and $\setH_{\tilde\bS_D^\ast, \tilde\bS_D} \hat f$ of $\hat f$. 
The objective is convex if the matrix
\begin{align}
	\setH \hat f = \begin{bmatrix}
		\setH_{\tilde\bS_D, \tilde\bS_D^\ast} \hat f & \setH_{\tilde\bS_D^\ast, \tilde\bS_D^\ast} \hat f \\
		\setH_{\tilde\bS_D, \tilde\bS_D} \hat f& \setH_{\tilde\bS_D^\ast, \tilde\bS_D} \hat f
	\end{bmatrix}
\end{align}
is positive semidefinite. Recognizing that the \mbox{second-order} derivatives of the constant 
and affine 
terms in \eqref{eq:traceform}
are zero, and following \cite[Ex. 5.4]{hjorungnes11a}, we obtain
\begin{align}
	\setH \hat f = \begin{bmatrix}
 	\tp{[\herm{\hat\bH}\tilde\bP\hat\bH\!-\!\alpha \bI_U]}\kron \bI_U & \mathbf{0} \\
 	\mathbf{0} & [\herm{\hat\bH}\tilde\bP\hat\bH\!-\!\alpha \bI_U] \kron \bI_U
 \end{bmatrix}. \!\!
\end{align}
$\setH \hat f$ is positive semidefinite if all its eigenvalues are non-negative. 
The eigenvalues of $\setH \hat f$ are simply the eigenvalues of $\herm{\hat\bH}\tilde\bP\hat\bH\!-\!\alpha \bI_U$, 
and they are all non-negative as long as $\alpha\leq \lambda_\text{min}$, 
where $\lambda_\text{min}$ is the smallest eigenvalue of $\herm{\hat\bH}\tilde\bP\hat\bH$. 

We now turn to the minimization with respect to $\tilde\bP$, and define $\bE \triangleq \bY\!_D - \hat\bH\tilde\bS_D$.
Note that any $\tilde\bP\in \mathscr{G}_{B-I}(\opC^B)$
can be written as $\tilde\bP = \bI_B - \bQ\herm{\bQ}$, where $\bQ\in\opC^{B\times I}$ consists of $I$ 
orthonormal columns. The second term of \eqref{eq:sandman_problem} does not depend on $\tilde\bP$, and since the squaring 
of the norm does not affect the minimizing argument (which is what we are interested in), it suffices to consider
\begin{align}
	\min_{\tilde\bP \,\in\, \mathscr{G}_{B-I}(\opC^B)} \| \tilde\bP\bE \|_F
	&= \min_{\bQ} \| \bE - \bQ\herm{\bQ}\bE \|_F \label{eq:Q_range} \\ 
	&\geq \min_{\tilde\bE\,:\,\rank{\tilde\bE}\leq I} \|\bE - \tilde\bE\|_F. \label{eq:E_range}
\end{align}
Here, \eqref{eq:E_range} follows since the matrix $\bQ\herm{\bQ}\bE$ has at most rank $I$ (since $\bQ$ has dimensions $B\times I$),
so the optimization range in \eqref{eq:Q_range} is a subset of the optimization range in \eqref{eq:E_range}. 
By the Eckart-Young-Mirksy theorem, \eqref{eq:E_range} is minimized for $\tilde\bE = \bU_I \boldsymbol{\Sigma}_I \herm{\bV_I}$,
where $\boldsymbol{\Sigma}_I = \text{diag}(\sigma_1, ..., \sigma_I)$ is the diagonal matrix 
whose diagonal entries are the $I$ largest singular values of $\bE$, and $\bU_I$ and $\bV_I$ consist of the corresponding
left- and right-singular vectors, respectively \cite{eckart1936approximation}. 
But if we let $\bE=\bU\boldsymbol{\Sigma}\herm{\bV}$ be the singular-value decomposition of $\bE$ and choose
$\bQ=\bU_I$, then $\bQ\herm{\bQ}\bE = \bU_I \boldsymbol{\Sigma}_I \herm{\bV_I}$, meaning that \eqref{eq:E_range} holds with 
equality and that \eqref{eq:Q_range} is minimized for $\bP=\bI_B - \bU_I \herm{\bU_I}$. 
\hfill$\blacksquare$


\end{document}